%% file: main.tex
\documentclass[conference]{IEEEtran}
\IEEEoverridecommandlockouts
\input{newmacro.tex}
\usepackage{cite}
\usepackage{amsmath,amssymb,amsfonts}
\usepackage{graphicx}
\usepackage{textcomp}
\usepackage{amssymb,amsmath,amsthm}

\newtheorem{theorem}{Theorem}
\newtheorem{definition}{Definition}
\newtheorem{lemma}{Lemma}

\usepackage{xcolor}
\def\BibTeX{{\rm B\kern-.05em{\sc i\kern-.025em b}\kern-.08em
    T\kern-.1667em\lower.7ex\hbox{E}\kern-.125emX}}
\begin{document}

\title{A Rademacher Complexity Based Method for Controlling Power and Confidence Level in Adaptive Statistical Analysis\thanks{This paper was submitted to IEEE/ACM/ASA DSAA 2019 on 05/20/2019 and accepted on 07/26/2019.
This research was funded by NSF Award RI-18134446, and by DARPA/USAF grant W911NF-16-1-0553.}}

\author{\IEEEauthorblockN{Lorenzo De Stefani}
\IEEEauthorblockA{\textit{Department of Computer Science} \\
\textit{Brown University}\\
Providence, United States of America \\
lorenzo@cs.brown.edu}
\and
\IEEEauthorblockN{Eli Upfal}
\IEEEauthorblockA{\textit{Department of Computer Science} \\
\textit{Brown University}\\
Providence, United States of America \\
eli@cs.brown.edu}
}

\maketitle
\begin{abstract}
\input{abstract.tex}
\end{abstract}


\input{01_introduction}
\input{02_radabound}
\input{03n_boundinggenerror}
\input{04_algo}
\input{05_experiments}
\input{07_diffprivacy}

\input{06_conclusion}

\bibliographystyle{IEEEtran}
\bibliography{holdout}

\end{document}

%% file: newmacro.tex
\newcommand{\algo}{\textsc{RadaBound}}

\newcommand{\Prob}[1]{\textrm{Pr}\left(#1\right)}
\newcommand{\Probd}[2]{\textrm{Pr}_{#1}\left(#2\right)}
\newcommand{\Exp}[1]{\text{E}\left[#1\right]}
\newcommand{\Expe}[2]{\text{E}_{#1}\left[#2\right]}
\newcommand{\Var}[1]{\text{Var}\left[#1\right]}

\newcommand\abs[1]{\left|#1\right|}

\newcommand{\Rade}{Rademacher Complexity}

\newcommand{\maxErr}[1]{\Psi\left(#1\right)}

\newcommand{\ExpeApp}[2]{\tilde{\text{E}}_{#1}\left[#2\right]}

\usepackage{amsmath}
\usepackage{mathrsfs}
\usepackage{algorithm}
\floatname{algorithm}{ALGORITHM}
\usepackage[noend]{algpseudocode}
\usepackage{varwidth}

\usepackage{graphicx}
\usepackage{subcaption}

\usepackage{pgfplots}
\usepackage{pgfmath}
\tikzset{
    declare function={
        normcdf(\x,\m,\s)=1/(1 + exp(-0.07056*((\x-\m)/\s)^3 - 1.5976*(\x-\m)/\s));
    }
}

%% file: abstract.tex
 While standard statistical inference techniques and machine learning generalization bounds assume that tests are run on data selected independently of the hypotheses, practical data analysis and machine learning are usually iterative and adaptive processes where the same holdout data is often used for testing a sequence of hypotheses (or models), which may each depend on the outcome of the previous tests on the same data. 
In this work, we present \algo{} a rigorous, efficient and practical procedure for controlling the generalization error when using a holdout sample for multiple adaptive testing. Our solution is based on a new application of the \Rade{} generalization bounds, adapted to dependent tests. We demonstrate the statistical power and practicality of our method through extensive simulations and comparisons to alternative approaches. In particular, we show that our rigorous solution is a substantially more powerful and efficient than the differential privacy based approach proposed in Dwork  et al.~\cite{Dwork15generalization,dwork2015reusable,dwork2015preserving}.


%% file: 01_introduction.tex
\section{Introduction}
The goal of data analysis and statistical learning is to model a stochastic process, or distribution, that explain an observed data. A major risk in statistical learning is \emph{overfitting}, that is, learning a model that fits well with the observed data but does not predict new data. The standard practice in machine learning is to split the data into \emph{training} and \emph{holdout} (or testing) sets. A learning algorithm then learns a model using the training data and tests the model on the holdout set to obtain a confidence interval for the expected error or for the value of the loss function of the model. If the process halts after a single iteration, then the statistical analysis is easy. However, in most cases, the learning process is iterative and adaptive. One uses successive tests for model selection, feature selection, parameter tuning, etc., and the choice of the tests themselves often depends on the outcomes of previous tests. 
Ideally, each hypothesis should be tested on a fresh data sample. However, it is common practice to reuse the same holdout data to evaluate a sequence of hypotheses. While widespread, this practice is known to lead to \emph{overfitting}; that is, the learned model becomes representative of the sample rather than the actual process. Evaluating the accumulated error in testing a sequence of related hypothesis on the same data set is a major challenge in both machine learning and modern statistics. In machine learning, the problem of ``\emph{preventing overfit}'', is usually phrased and analyzed in terms of bounding the generalization error~\cite{ShalevSBD14}. In inference statistics, the goal is controlling the Family Wise Error Rate (FWER), or the False Discovery Rate (FDR) of a sequence of hypothesis tests~\cite{benjamini1995controlling}.\\

\noindent\textbf{Our Results:} We develop and analyze \algo{}, a rigorous, efficient, and practical procedure for online evaluation of the accumulated generalization error in a sequence of statistical inferences applied to the same sample. \algo{}  can evaluate fully adaptive sequences of tests. The choice of a test may depend on the information obtained from previous tests, and the total number of tests is not fixed in advance. 

One way to quantify the risk of overfitting after $k$ queries is by considering the probability of the condition defined by the results of the first $k$ queries. If the probability of such condition is close to 1, the results of the queries evaluated so far do not significantly restrict the sample space, and there is, therefore, no risk of overfitting. Viceversa, if the probability of the observed condition is small, the sample space defined by the true distribution conditioned on the results of the queries is noticeably different from the true distribution, and there is thus a significant risk of overfitting. In general, it is hard to bound the probability of the observed condition as the true distribution over the samples is unknown. However, in the special case for which the queries being considered correspond to evaluating the average of functions (such as evaluating the average risk or loss functions of alternative learning procedures), we can design an adaptive process based on an empirical estimate of the Rademacher Complexity of the set of queries which correctly bounds this probability and correspondingly halts the procedure when the risk of overfitting exceeds a certain threshold fixed by the user.

Our method builds on the concept of Rademacher Complexity~\cite{BartlettBL02,Koltchinskii01} that has emerged as a
powerful alternative to VC-dimension and related  uniform convergence methods for
characterizing generalization error and sample complexity. A fundamental advantage of the Rademacher Complexity approach in contrast to standard uniform convergence tools, such as
 VC-dimension, that capture the complexity with respect the worst case input
 distribution, is that it yields a data-dependent bound as it
is computed with respect to the input (sample) distribution, and can be efficiently approximated from the sample. 

Our solution employs three major components: (1) For a set of functions chosen independent of the sample, the \Rade{}~\cite{BartlettBL02,Koltchinskii01} provides a powerful and efficient bound on the error in estimating the expectations of all these function using one sample; (2) As long as the outcome of the sequence of tests does not significantly overfit to the sample, conditioning on these outcomes has only a minor effect on the distribution; and (3) The \Rade{} of a sequence of tests can be estimated efficiently form a given sample, requiring similar computation time as running the actual tests.  To fully utilize our technique, we need computationally efficient methods for rigorously estimating the Rademacher Complexity from a sample. We introduce two novel methods based on Bernstein's inequality for martingales~\cite{freedman1975tail} and the Martingale Central Limit Theorem~\cite{hall2014martingale}. 
Our analysis and extensive experiments prove and demonstrate that  our method guarantees statistical validity while retaining statistical power and practical efficiency.\\

\noindent\textbf{Related Work:} Classic statistics offers a variety of procedures for controlling the Family Wise Error Rate (FWER), ranging from the simple Bonferroni~\cite{bonferroni1936teoria} to Holm's step-down~\cite{holm1979simple} and Hochberg's step-up procedures~\cite{hochberg1988sharper} in the context of multiple hypotheses testing. While controlling the FWER under weak assumptions about the hypotheses, these methods are too conservative, giving many false negative results, in particular for large sets of hypotheses. Less conservative procedures, such as Benjamini and Hochberg~\cite{benjamini1995controlling}, which control the False Discovery Rate (FDR) (i.e., the expected fraction of false discoveries), still do not scale up well for a very large number of hypotheses. However, all these procedures cannot be applied in the adaptive setting, as they require for the set of hypotheses to be fixed at the beginning of the testing procedure (i.e., before any data evaluation).

In statistics, ``\emph{sequential analysis}'' or ``\emph{sequential hypothesis testing}'' is a paradigm for statistical testing where for a fixed family of hypotheses to be the tested  the sample size is not fixed in advance. Instead, data are evaluated as they are collected, and further sampling is stopped in accordance with a pre-defined stopping rule as soon as significant results are observed. Despite the sequential iterative nature of this practices, as the hypotheses being considered are fixed beforehand,  sequential analysis procedures are not suitable for adaptive analysis as the set of queries (hypotheses) being considered depends in general for the outcome of previous evaluation of the data itself. Other ``\emph{sequential}'' hypothesis testing procedures, such as the  \emph{sequential False Discovery Rate} control by G'Sell et al.~\cite{g2016sequential}, assume that the order according to which the hypotheses are to be evaluated is fixed beforehand, and hence cannot be adaptively selected. Similar considerations apply to the ``\emph{Alpha Investing}'' sequential testing by Foster and Stine~\cite{foster2008alpha}, which achieves control of the ``\emph{marginal False Discovery Rate}''. While the previously mentioned procedures apply to the setting of hypotheses testing, the method proposed in this work allows adaptive evaluations of statistical queries while maintaining rigorous guarantees on the accuracy of the obtained estimates.

A series of recent papers~\cite{Dwork15generalization,dwork2015reusable,dwork2015preserving} explored an interesting relation between ``\emph{Differential Privacy}''~\cite{dwork2006calibrating} and overfit prevention in adaptive analysis. The basic idea is to limit the user access to the holdout data so that the answers to the sequence of queries is differentially private. A differentially private access to the holdout data limits the risk of overfitting to that data set. Unfortunately, the practical application of this elegant mathematics is limited. Differential privacy is achieved through random perturbation of the data (or the reply to the queries). The higher the number of adaptive queries, the larger the required perturbation. However, the amount of perturbation is limited by the need to preserve the actual signal in the data. As a result, rigorous application of this approach is either limited to a small number of queries or is computationally intractable~\cite{dwork2015preserving}, making it less useful than alternative methods~\cite{hardt2014preventing,steinke2015interactive}. 
Our experiments in Section 5 show that \algo{} allows orders of magnitude reduction of the required holdout dataset compared to Dwork et al.'s method~\cite{Dwork15generalization} while offering the same guarantees. Further, our technique is much simpler as it does not require any introduction of additional noise.
We discuss in detail the advantages of our solution compared to ~\cite{Dwork15generalization} in Section~\ref{sec:dwork}.
A more practical solution for a restricted setting inspired by machine learning competitions was presented in~\cite{blum2015ladder}. Their solution, ``\emph{the Ladder}'', 
provides a loss estimate only for those that made a significant improvement over the previous best. This restricted setting allows to sidestep the hardness results discussed in~\cite{hardt2014preventing,steinke2015interactive}. Note however that the guarantee achieved by the Ladder is fundamentally different from the one achieved in Dwork et al.~\cite{dwork2015reusable} and in this work, as it ensures accuracy in the relative ordering of the performance of multiple classifiers, while the latter ensure accurate evaluations of adaptively selected queries. Hence, the former is not comparable with the latters.\\ 

\noindent\textbf{Paper organization:} The presentation is organized as follows:
In Section II we introduce out \texttt{RadaBound} method for adaptive statistical analysis. In Section III we discuss the use of uniform convergence bounds based on Rademacher Complexity in our setting, and we present two methods for estimating the Rademacher Complexity of a class of adaptively selected functions from the data.
In Section IV we present the details of our  \texttt{RadaBound} and the guarantees provided by it. In Section V, we present an experimental validation of the correctness and power of our methods using synthetic data. Finally, in Section VI we compare our approach to the state-of-the-art approach based on \emph{Differential Privacy} by Dwork et al.~\cite{dwork2015reusable}. In particular, we show that a Rademacher complexity based solution gives significantly better results than the more complicated differential privacy based solution of Dwork et al..

%% file: 02_radabound.tex
\section{The \textsc{RadaBound}}
For concreteness, we focus on the following setup.
We have an holdout sample composed by $m$ independent observations $\bar{x}=(x_1,\dots,x_m)$, each from a distribution $\mathcal{D}$, and parameters $\epsilon,\delta\in (0,1)$ fixed by the user.\\

\noindent\textbf{The process:} In an iterative process, at each step, the user (or an adaptive algorithm) submits a function $f$ and receives an estimate $\ExpeApp{\bar{x}}{f}=\frac{1}{m}\sum_{i=1}^m f(x_i)$  of the ``\emph{ground truth value}'' $\Expe{\mathcal{D}}{f}$.
The user has no \emph{direct access} to the sample $\bar{x}$. That is, he can \emph{only} acquire information regarding $\bar{x}$ from the confidence intervals $\ExpeApp{\bar{x}}{f}\pm\epsilon$ for the expectation $\Expe{\mathcal{D}}{f}$, which the testing procedure has returned as answer to the queries considered so far.
Let $\mathcal{F}_k=\{f_1,\dots, f_k\}$ denote the set of the first $k$ functions evaluated during the adaptive process. The maximum error in estimating the expectations of the $k$ functions is given by: 
\begin{align*}
\Psi \left(\mathcal{F}_k,\bar{x}\right) &= \sup_{f\in\mathcal{ F}_k} |\frac{1}{m}\sum_{i=1}^m f(x_i)- \Expe{\mathcal{ D}}{f}|\\ &= \sup_{f\in\mathcal{ F}_k} |\ExpeApp{\bar{x}}{f}- \Expe{\mathcal{D}}{f}|.
\end{align*}
In this work, we use the expression ``\emph{overfittig}'' as follows: A given set of functions $\mathcal{F}_k$ is said to \emph{overfit the sample} $\bar{x}$ if for any $f\in \mathcal{F}$ the value  $\ExpeApp{\bar{x}}{f}$ evaluated on $\bar{x}$ differs from the \emph{true value} $\Expe{\mathcal{D}}{f}$ by more than the user given threshold $\epsilon$. Our adaptive testing process halts at the first $k$-th step for which for which it cannot guarantee that the probability of overfitting is at most $\delta$, that is, when  $\Prob{\maxErr{\mathcal{F}_k,\bar{x}}\leq \epsilon}\geq 1 -\delta$.

The process is fully adaptive. The choice of the function $f_{k+1}$ evaluated at the $k+1$-th step may depend on the information obtained during the first $k$ steps. We make no assumptions on the processes according to which the functions are adaptively chosen to be tested, nor do we require the total number of tests to be fixed in advance. For simplicity, we assume that all functions are in the range $[0,1]$. More general settings are discussed later in the paper.\\

\noindent\textbf{Bounding the generalization error for the iterative process:} The sequence of answers to the queries, $\ExpeApp{\bar{x}}{f_1}\pm\epsilon,\ExpeApp{\bar{x}}{f_2}\pm\epsilon, \dots,\ExpeApp{\bar{x}}{f_k}\pm\epsilon$ defines a filtration 
$\mathcal{L}=\{\mathcal{ D}_k\}_{k\geq 0}$, such that 
$$\mathcal{ D}_0= \mathcal{ D}~~ \mbox{and}~~  \mathcal{ D}_k= \{\mathcal{ D}~|~ \ExpeApp{\bar{x}}{f_1}\pm\epsilon \wedge \dots \wedge \ExpeApp{\bar{x}}{f_k}\pm\epsilon\}.$$ The $k$-th query is chosen with respect to, and is answered in the filtered distribution $\mathcal{D}_{k-1}$.
Our first step in developing \algo{} is to adapt the Rademacher Complexity results to an iterative, adaptive sequence of queries. 

Let $E_k$ denote the event that the answer to the $k$-th query was within $\epsilon$ of the correct value, that is,  $E_k := |\ExpeApp{\bar{x}}{f}- \Expe{\mathcal{D}}{f_k}|\leq \epsilon$.
Then,
$\Prob{\maxErr{\mathcal{ F}_k,\bar{x}}\leq \epsilon}=\Prob{\wedge_{i=1}^k E_k}$, and thus, in the filtration process,
\begin{align*}\label{eq:iterativeproc}
&\Probd{\mathcal{L}}{\Psi (\mathcal{ F}_k,\bar{x})> \epsilon}\\
&\qquad \leq \Probd{\mathcal{ D}_0}{\bar{E}_1} + \Probd{\mathcal{ D}_1}{\bar{E}_2} + \ldots 
 + \Probd{{\mathcal{ D}_{k-1}}}{\bar{E}_k}\nonumber \\ 
&\qquad= \sum_{i=1}^k \frac{\Prob{\bar{E}_i
\wedge (\wedge_{j=1}^{i-1}E_j)}}{\Prob{\wedge_{j=1}^{i-1}E_j}} \nonumber\\ 
&\qquad\leq \frac{1-\Prob{\wedge_{j=1}^{k}E_j}} {\Prob{\wedge_{j=1}^{k-1}E_j}}.
\end{align*}
By the definition of the events $E_j$ we thus have:
\begin{equation}\label{eq:iterativeproc}
    \Probd{\mathcal{L}}{\Psi (\mathcal{ F}_k,\bar{x})> \epsilon}\leq \frac{1- \Prob{\maxErr{\mathcal{ F}_k,\bar{x}}\leq \epsilon}}{\Prob{\Psi (\mathcal{ F}_{k-1},\bar{x})\leq \epsilon}}
\end{equation}

where $\Prob{}{}$ with no subscript refers to probability in the un-filtered distribution $\mathcal{D}$.

The fact that the distribution of the generalization error in the adaptive case, $\Probd{\mathcal{L}}{\maxErr{\mathcal{ F}_k,\bar{x}}> \epsilon}$, is related to the probability of an error in the non-adaptive case, $\Prob{\maxErr{\mathcal{ F}_k,\bar{x}}> \epsilon}$, is not surprising. In order to fit the sample differently than the original distribution $\mathcal{D}$, the process needs to detect a pattern whose frequency is considerably different in the sample compared to the actual distribution $\mathcal{D}$. However, the first query that observes such a pattern is chosen when the process has not yet observed a significant difference between the sample and the distribution. 
This is due to the fact that the process halts as soon as such difference is detected. Thus, the probability of overfitting in $k$ queries is related to the probability that the sample gives a bad estimate for the correct value of one of the $k$ queries in the non-adaptive case.

The challenge is to compute a tight bound to the probability $\Prob{\maxErr{\mathcal{ F}_k,\bar{x}}\leq \epsilon}$. We achieve this through two novel bounds on estimating the Rademacher Complexity of $\mathcal{ F}_k$.\\

%
%
\section{Bounding $\maxErr{\mathcal{ F}_k,\bar{x}}$ using \Rade{}}
Our solution is based on iterative applications of Rademacher Complexity bounds.
\begin{definition}{\cite{mitzenmacher2017probability}}  Let $\bar{\sigma}=(\sigma_1,\dots,\sigma_m)$ be a vector of $m$ independent \emph{Rademacher random variables}, such that for all $i$,  $\Prob{\sigma_i=1}=\Prob{\sigma_i=-1}=1/2$. 
The {\em Empirical Rademacher Complexity} of a class of function $\mathcal{
F}$ with respect to a sample $\bar{x}=\{x_1,\dots,x_m\}$, with $\bar{x}\sim \mathcal{D}^m$ is 
\begin{equation*}
R^{\mathcal {F}}_{\bar{x}} = \Expe{\bar{\sigma}}{\sup_{f\in \mathcal{
    F}}\frac{1}{m} \sum_{i=1}^m f (x_i)\sigma_i }
\end{equation*}
The \emph{Rademacher Complexity} of $\mathcal{
F}$ for samples of size $m$ is defined as $R^{\mathcal{F}}_m:=\Expe{\bar{x}\sim\mathcal{ D}^m}{R^{\mathcal{F}}_{\bar{x}}}$.
    \end{definition}

The relation between $\Psi (\mathcal{ F}_k,\bar{x})$ and the Rademacher Complexity of $\mathcal{ F}_k$ is given by the following results~\footnote{In our setting, in order apply the result with absolute value we assume that for any $f\in\mathcal{ F}_k$ we also have $-f\in \mathcal{ F}_k$, i.e., we assume that $\mathcal{F}_k$ is \emph{closed under negation}.}:
\begin{lemma}[\emph{Lemma 26.2,} \cite{ShalevSBD14}] 
\label{th:Ra}
\begin{align*}
\Expe{\bar{x}\sim\mathcal{ D}^m}{\Psi (\mathcal{ F}_k,\bar{x})} &=\Exp{\sup_{f\in\mathcal{ F}_k} |\frac{1}{m}\sum_{i=1}^m f(x_i)- E_\mathcal{ D}[f]|}\\
&\leq 2 R^{\mathcal{ F}_k}_m.
\end{align*}
\end{lemma}
\begin{lemma}[\emph{Theorem 14.21,} \cite{mitzenmacher2017probability}]\label{th:Ra2}
Assume that for all $x\in \mathcal{X}$ and $f\in \mathcal{F}_k$ we have $f(x)\in \left[0, 1 \right]$, then:
\begin{equation}\label{Radbound11}
\Prob{\Psi (\mathcal{ F}_k,\bar{x}) > 2R^{\mathcal{F}_k}_m+ \epsilon}\leq e^{-2m\epsilon^2}. 
\end{equation}
\end{lemma}
Note that 
in our context (a) we need a one-sided bound, and (b)  for all $x\in \mathcal{X}$ and $f\in \mathcal{F}_k$ we have $f(x)\in \left[0, 1 \right]$.

For algorithmic applications, two important consequences of these result are that: (a) for bounded functions the generalization error is concentrated around their expectation, and (b) the \Rade{} can be estimated from the sample.
In order for this bound to be actually usable in practical applications, it is necessary to compute an estimate of the \Rade{} given the dataset $\bar{x}$, and to bound its error. In the ``\emph{textbook}'' treatment, the difference between Rademacher Complexity and its empirical counterpart is bounded using a second application of \emph{McDiarmid's Inequality}~\cite{Barlett02,ShalevSBD14}. However, this bound is often too loose for practical applications such as ours. 

In this work, we propose an alternative, \emph{direct}, estimate of the \Rade{} and we develop two methods for tightly bounding the estimation error.

%% file: 03n_boundinggenerror.tex
\subsection{Tight bounds on \Rade{} estimate}\label{sec:bound}
Given a finite size sample $\bar{x}\sim \mathcal{D}^m$ and $\ell$ independent Rademacher vectors  $\bar{\sigma}_1,\dots, \bar{\sigma}_\ell$, each composed of $m$ independent Rademacher random variables (i.e., $\bar{\sigma}_j = \sigma_{j,1}\sigma_{i,2}\ldots\sigma_{j,m}$ ), we estimate ${R}_m^{\mathcal{F}_k}$ with
%
\begin{equation}\label{eq:radestimate}	\tilde{R}_{\bar{x},\ell}^{\mathcal{F}_k}=\frac{1}{\ell}\sum_{j=1}^\ell\sup_{f\in \mathcal{F}_k}\frac{1}{m}\sum_{i=1}^m f(x_i)\sigma_{j,i}.
\end{equation}
Clearly, $\Expe{\bar{x},\bar{\sigma_1},\dots,\bar{\sigma}_\ell} {\tilde{R}_{\bar{x},\ell}^{\mathcal{F}_k}}=R^{\mathcal{F}_k}_m$.
To bound the error $R^{\mathcal{F}_k}_m- \tilde{R}_{\bar{x},\ell}^{\mathcal{F}_k}$, we model the process as a \emph{Doob martingale} (\cite[Chapter 13.1]{mitzenmacher2017probability}) as follows: 
$$C_i=E[R^{\mathcal{F}_k}_m- \tilde{R}_{\bar{x},\ell}^{\mathcal{F}_k}~|~Y_1,\dots,Y_i]~~ \mbox{for}~i=0,\dots,m(\ell+1),$$
where the $Y_1,\ldots,Y_{m(\ell+1)}$ are the random variables that determinate the value of the estimate $\tilde{R}_{\bar{x},\ell}^{\mathcal{F}_k}$. The first $m$ variables $Y_i$'s correspond to the values of the sample $\bar{x}$, i.e. for $1\leq i\leq m$, $Y_i=X_i$, and the remaining $m\ell$ $Y_i$'s correspond to the Rademacher random variables,  $Y_i= \sigma_{\lfloor i/m \rfloor, i- \lfloor i/m \rfloor}$. It is easy to verify that $C_0=0$, and $C_{m(\ell+1)}= R^{\mathcal{ F}}_m- \tilde{R}_{\bar{x},\ell}^{\mathcal{F}_k}$.

Next, we define a \emph{martingale difference sequence} $Z_i=C_i-C_{i-1}$ with respect to the martingale $C_0, C_1, \dots C_{m(\ell+1)}$, and note that $\sum_{t=1}^{m(\ell+1)} Z_t=C_{m(\ell+1)} = R^{\mathcal{F}_k}_m- \tilde{R}_{\bar{x},\ell}^{\mathcal{F}_k}$.\\
%

\noindent\textbf{Application of Bernstein's Inequality for Martingales:} Our first bound builds on Bernstein's Inequality for Martingales (BIM). We use the following version due to  Freedman~\cite{freedman1975tail}, as presented in~\cite{dzhaparidze2001bernstein} and adapted to one-sided error.
\begin{theorem}\label{thm:bernmain}{\cite{dzhaparidze2001bernstein,  freedman1975tail}}
Let $Z_1,\ldots,Z_t $ be a martingale difference sequence with respect to a certain filtration $\{\mathscr{F}_i\}_{i=0,\ldots,t}$. 

Thus, $\Exp{Z_i|\mathscr{F}_{i-1}} = 0$ for $i=1,\ldots,t$. The process $\sum_{i=1}^t Z_i$ is thus a martingale with respect to this filtration. Further, assume that $|Z_i|\leq a$ for $i=1,\ldots,t$, and that the conditional variance $\sum_{i=1}^t \Exp{Z_i^2}\leq L$. For $\epsilon \in (0,1)$, we have:
\begin{equation}
\Prob{\sum_{i=1}^t Z_i>\epsilon}\leq e^{-\frac{\epsilon^2}{2L+2a\epsilon/3}}.
\end{equation}
\end{theorem}
Note that the bound presented here is slightly different from the one in~\cite{freedman1975tail} as for our purposes we only require a one-sided bound. A careful analysis of $E[Z_i^2]$ in our application allows us to obtain a significantly stronger bound than the one obtained using McDiarmid's Inequality~\cite{Barlett02,ShalevSBD14}, which depends on the maximum variation of the martingale. 
\begin{theorem}
Given a sample $\bar{x}\sim \mathcal{D}^m$, a family of functions $\mathcal{F}_k$ which take values in $[0,1]$,  $\ell$ independent vectors of Rademacher random variables, and $\epsilon, \delta \in (0,1)$, we have: 
\label{lem:bernrade}
\begin{equation}
\label{eq:radestimate1}
\Prob{R^{\mathcal{F}_k}_{m}- \tilde{R}^{\mathcal{F}_k}_{\bar{x},\ell}>\epsilon}\leq  e^{- \frac{6m \ell \epsilon^2}{15+ 8\ell \epsilon}}.
 \end{equation}
\end{theorem}
\begin{IEEEproof}
Recall the definition of the Doob martingale $$C_i=E[R^{\mathcal{F}_k}_m- \tilde{R}_{\bar{x},\ell}^{\mathcal{F}_k}~|~Y_1,\dots,Y_i],$$  for $i=0,\dots,m(\ell+1)$, and the the definition of the corresponding martingale difference sequence $Z_i = C_{i}-C_{i-1}$.

By definition, for every $i=1,\ldots,m(\ell+1)$, we have $\Exp{Z_i}=0$, and hence, $\Exp{Z_i^2}= \Var{Z_i}$.  

In order to apply Bernstein's Inequality, we need a bound $a$, such that $a\geq |Z_i|$ for $1\leq i \leq m(\ell+1)$, and an upper-bound $L$ to the conditional variance, such that $$L\geq \sum_{i=1}^{m(\ell+1)}\Exp{Z_i^2} = \sum_{i=1}^{m(\ell+1)}\Var{Z_i}.$$

We consider the cases for $1\leq i\leq m$ and $m<i\leq m(\ell+1)$ separately:
\begin{itemize}
\item $1\leq i\leq m$: For $1\leq j \leq \ell$, let us consider 
\begin{align*}
	C_i^{(j)} &= E[R^{\mathcal{F}_k}_m- \sup_{f\in \mathcal{F}_k}\frac{1}{m}\sum_{i=1}^m f(x_i)\sigma_{j,i}~|~Y_1,\dots,Y_i],\\
	Z_i^{(j)} &= C_{i}^{(j)} - C_{i-1}^{(j)}.
\end{align*}
According to our definitions, we have
\begin{align*}
	C_i &= \frac{1}{\ell}\sum_{j=1}^\ell C_i^{(j)}\\
    Z_i &= \frac{1}{\ell}\sum_{j=1}^\ell Z_i^{(j)}.
\end{align*}
Since $\forall x\in \mathcal{X}$ and $\forall f \in \mathcal{F}_k$, $f(x) \in [0,1]$, changing the value of any of the $m$ points in $\bar{x}$ can change $\sup_{f\in \mathcal{F}_k}\frac{1}{m}\sum_{i=1}^m f(x_i)$ by at most $1/m$, and thus we have $|Z_i^{(j)}|\leq 1/m$, and $Z_i^{(j)}\in[\alpha,\beta]$ with $\beta - \alpha \leq 1/m$.

From \emph{Popoviciu's Inequality on variance}~\cite{popoviciu1935equations}, we have that the variance of a random variable which takes values in $[\alpha, \beta]$ is bounded from above by $(\beta-\alpha)^2/4$. Hence,
by applying Popoviciu's Inequality to $Z_i^{(j)}$, we have that $\Var{Z_i^{(j)}}\leq 1/(4m^2)$. 

As we are considering the expectation over the unassigned values of the Rademacher random variables, and as we are averaging over the values obtained using $\ell$ independent and identically distributed vectors of Rademacher random variables, we can conclude that $|Z_i|\leq \frac{1}{\ell}\sum_{j=1}^\ell|Z_i^{(j)}| \leq 1/m$, and  $\Var{Z_i} = \frac{1}{\ell^2}\sum_{j=1}^\ell \Var{Z_i^{j}} \leq 1/(4m^2\ell)$.
\item $ m<i\leq m(\ell+1)$: Changing the value of any of the $\ell m$ Rademacher random variables can change the value of $\tilde{R}^{\mathcal{F}_k}_{\bar{x},j}$ by at most $2/\ell m$, and thus we have $|Z_i|\leq 2/\ell m \leq 1/m$, and $Z_i\in[\alpha,\beta]$ with $\beta - \alpha \leq 2/\ell m$.
By applying Popoviciu's Inequality, we thus have $\Var{Z_i}\leq 1/\ell^2 m^2$.
\end{itemize}
By linearity of expectation, $\sum_{i=1}^{m(\ell+1)} Z_i = R^{\mathcal{F}_k}_{m} - \tilde{R}^{\mathcal{F}_k}_{\bar{x},\ell}$. Further, we have $\sum_{i=1}^{m(\ell+1)} Z_i^2 \leq 5/4\ell m$, and $|Z_i|<1/m$ for all $1\leq i\leq m(\ell+1)$.
The statement follows by applying Theorem~\ref{thm:bernmain}. 
\end{IEEEproof}

Note that for a sufficiently large (constant) $\ell$, the term $6\epsilon \ell$ dominates the denominator of the exponent in the right hand side of~\eqref{eq:radestimate1}, giving
a \emph{fast rate of convergence} for the estimate. Our estimate fully characterizes the benefit achieved using  multiple independent vectors of Rademacher random variables in estimating $\tilde{R}^{\mathcal{F}_k}_{\bar{x},\ell}$. 


Combining the results of Theorem~\ref{lem:bernrade}, Lemma~\ref{th:Ra}, and Lemma~\ref{th:Ra2} using the union bound, we obtain an empirical bound on $\maxErr{\mathcal{F}_k,\bar{x}}$.
\begin{theorem}\label{thm:bernradetot}Given a sample $\bar{x}\sim \mathcal{D}^m$, a family of functions $\mathcal{F}_k$ which take values in $[0,1]$,  $\ell$ independent vectors of Rademacher random variables, and $\epsilon, \delta \in (0,1)$, we have:
\small
\begin{equation}
\Prob{\Psi (\mathcal{F}_k,\bar{x}) > 2\tilde{R}_{\bar{x},\ell}^{\mathcal{F}_k}+ \epsilon}  
< \min_{\alpha\in(0,\epsilon)}  e^{-2m(\epsilon-\alpha)^2}+e^{-\frac{3m \ell \alpha^2}{30+ 8\ell \alpha}}.
\end{equation}
\normalsize
\end{theorem}
\begin{IEEEproof}
From Lemmas~\ref{th:Ra} and~\ref{th:Ra2}, we have:
\begin{equation*}
\Prob{\Psi (\mathcal{ F}_k,\bar{x}) > 2R^{\mathcal{F}_k}_m+ \epsilon_1}\leq e^{-2m\epsilon_1^2}.
\end{equation*}
Theorem~\ref{lem:bernrade} characterizes the quality of the estimate of the \Rade{} given by $\tilde{R}_{\bar{x},\ell}^{\mathcal{F}_k}$, computed as specified in \eqref{eq:radestimate}:
\begin{equation*}
\Prob{R^{\mathcal F}_{m}- \tilde{R}^{\mathcal{F}_k}_{\bar{x},\ell}>\epsilon_2}\leq  e^{- \frac{6m \ell \epsilon_2^2}{15+ 8\ell \epsilon_2}}.
\end{equation*}
Combining the two results, we obtain: 
\begin{equation*}
\Prob{\Psi (\mathcal{ F}_k,\bar{x}) > 2\tilde{R}^{\mathcal{F}_k}_{\bar{x},\ell}+ \epsilon_1+2\epsilon_2}\leq e^{-2m\epsilon_1^2}+e^{- \frac{6m \ell \epsilon_2^2}{15+ 8\ell \epsilon_2}}.
\end{equation*}
By substituting $\alpha = 2\epsilon_2$ and $\epsilon = \epsilon_1 + 2\epsilon_2$ in the previous equation, we have:
\begin{equation*}
\Prob{\Psi (\mathcal{ F}_k,\bar{x}) > 2\tilde{R}^{\mathcal{F}_k}_{\bar{x},\ell}+ \epsilon}\leq e^{-2m(\epsilon - \alpha)^2}+e^{- \frac{6m \ell (\alpha/2)^2}{15+ 8\ell \epsilon_2}}.
\end{equation*}
The statement follows.
\end{IEEEproof}
\noindent\textbf{Alternative bound with single application of Bernstein's Inequality for Martingales:}
We now present an alternative result to the one in Theorem~\ref{thm:bernradetot}, which can be achieved with a single application of BIM. This bound is tighter than the one in Theorem~\ref{thm:bernradetot} when the number of independent vectors of Rademacher random variables is very high.

\begin{theorem}\label{thm:bern1bound}
	\begin{equation*}
		\Prob{\maxErr{\mathcal{F}_k, \bar{x}}> 2\tilde{R}_{\bar{x},\ell}^{\mathcal{F}_k}+ \epsilon}<e^{-\frac{\epsilon^2}{\frac{\ell + 4\sqrt{\ell}+20}{2m\ell}+\frac{4 \epsilon}{3m}}}
			\end{equation*}
\end{theorem}
\begin{IEEEproof}
	Consider the Doob supermartingale:
	\begin{equation*}
		C_i = \Exp{\maxErr{\mathcal{F}_k,\bar{x}} - 2\tilde{R}^{\mathcal{F}_k}_{\bar{x},\ell}|Y_1,\ldots Y_i}~~ \mbox{for}~i=0,\dots,m(\ell+1),
	\end{equation*}
where for $1\leq i\leq m$,
$Y_i = X_i$, and the remaining $Y_i$ correspond to the $m\ell$ independent Rademacher random variables in the $\ell$ vectors; that is, $Y_{j(m)+i}= \sigma_{j,i}$ for $1\leq j\leq \ell$ and $1\leq i \leq m$.
It is easy to verify that $C_{m(\ell+1)}= \maxErr{\mathcal{F}_k,\bar{x}}- 2\tilde{R}_{\bar{x},\ell}^{\mathcal{F}_k}$. Further,  $C_0=\Exp{\maxErr{\mathcal{F}_k,\bar{x}}} - 2R^{\mathcal{F}_k}_{m}$, and due to Theorem~\ref{th:Ra}, $C_0\leq 0$.

Let us define the corresponding martingale difference sequence $Z_i=C_{i}-C_{i-1}$.
For each $i\in\{1,\ldots,m(\ell+1)\}$, due to linearity of expectation, we have $Z_i = A_i - 2B_i$, where:
\begin{align*}
A_i &= \Exp{\maxErr{\mathcal{F}_k,\bar{x}}|Y_1,\ldots Y_i}- \Exp{\maxErr{\mathcal{F}_k,\bar{x}}|Y_1,\ldots Y_{i-1}};\\
B_i &= \Exp{\tilde{R}^{\mathcal{F}_k}_{\bar{x},\ell}|Y_1,\ldots Y_{i}}- \Exp{\tilde{R}^{\mathcal{F}_k}_{\bar{x},\ell}|Y_1,\ldots Y_{i-1}}.
\end{align*}
In order apply Bernstein's Inequality, we need an upper-bound $a\geq |Z_i|$ for $1\leq i \leq m (\ell+1)$ and an upper-bound $L$, such that  $L\geq  \sum_{i=1}^{m(\ell+1)}\Exp{Z_i^2}$.

Given our definition of $Z_i$,  we have that for every $i$, $\Exp{Z_i}=\Exp{A_i}=\Exp{B_i}=0$, and thus: $\Exp{Z_i^2}=\Var{Z_i}\leq \Var{A_i}+4\Var{B_i} +4\text{Cov}\left[A_i, B_i\right]$. From the properties of covariance, we have $|\text{Cov}\left[A_i, B_i\right]|\leq \sqrt{\Var{A_i}\Var{B_i}}$, and thus, $\Exp{Z_i^2}=\Var{Z_i}\leq \Var{A_i}+4\Var{B_i} + 4\sqrt{\Var{A_i}\Var{B_i}}$.

We consider the cases for $1\leq i\leq m$ and $m<i\leq m(\ell+1)$ separately:
\begin{itemize}
\item $1\leq i\leq m$: In our setting $\forall x\in \mathcal{X}$ and $\forall f \in \mathcal{F}$, $f(x) \in [0,1]$, changing the value of any of the $m$ points in $\bar{x}$ can change $f(\bar{x})$ by at most $1/m$. Therefore, $|A_i|\leq 1/m$, and $A_i\in[\alpha,\beta]$ with $\beta - \alpha \leq 1/m$. By applying Popoviciu's Inequality, we have: $\Var{A_i}\leq 1/4m^2$. 

The analysis for $\Var{B_i}$ follows the same reasoning discussed in the proof of Theorem~\ref{lem:bernrade} for bounding $\Var{Z_i}$ in the case $1\leq i\leq m$, and thus $\Var{B_i}\leq 1/(4m^2\ell)$. We can thus conclude:
 \begin{align*}
	\Exp{Z_i^2}&=\Var{Z_i} \leq \frac{1}{4m^2}+\frac{4}{4\ell m^2} + 4\sqrt{\frac{1}{4m^2}\frac{1}{4\ell m^2}}\\
    &\leq \frac{\ell + 4 +4\sqrt{\ell}}{4m^2\ell};\\
    |Z_i| &\leq \frac{2}{m}.
 \end{align*}

\item $m < i\leq m(\ell +1)$: Changing the value of any of the Rademacher random variables does not change the value of $\maxErr{\mathcal{F}_k,\bar{x}}$. Hence, $\Var{A_i}=\Exp{A_i^2} = 0$.

Given  fixed values for the random variables corresponding to the points in $\bar{x}$,  changing the value of one Rademacher random variable can change the value of $\tilde{R}^{\mathcal{F}_k}_{\bar{x},\ell}$ by at most $2/\ell m$. Thus, $|B_i| \leq \frac{2}{\ell m}$, and $B_i\in[\alpha,\beta]$ with $\beta - \alpha \leq 2/\ell m$. By applying Popoviciu's Inequality, we have:
\begin{align*}
	\Var{Z_i} = 4\Var{B_i} \leq \frac{4}{m^2 \ell^2}.
\end{align*}
\end{itemize}
We, therefore, have  $|Z_i| \leq \frac{2}{m}$ for all $1\leq i\leq m(\ell +1)$, and  $\sum_{i=1}^{m (\ell+1)}\Exp{Z_i^2} \leq \frac{\ell + 4\sqrt{\ell}+20}{4m\ell}$.
By linearity of expectation, and by applying Theorem~\ref{th:Ra}:
\begin{align*}
	\sum_{i=1}^{\ell (m+1)} Z_i &= \maxErr{\mathcal{F}_k,\bar{x}} - \Exp{\maxErr{\mathcal{F}_k,\bar{x}}} -2\left(\tilde{R}^{\mathcal{F}_k}_{\bar{x},\ell} - R^{\mathcal{F}_k}_{m}\right)\\
	&\geq \maxErr{\mathcal{F}_k,\bar{x}} - 2R^{\mathcal{F}_k}_{m} -2\left(\tilde{R}^{\mathcal{F}_k}_{\bar{x},\ell} - R^{\mathcal{F}_k}_{m}\right);\\
	&\geq \maxErr{\mathcal{F}_k,\bar{x}} - 2\tilde{R}^{\mathcal{F}_k}_{\bar{x},\ell};
\end{align*} 
The statement follows by applying BIM (Theorem 3.3).
\end{IEEEproof}

This result can be used in \algo{} in place of the bound given by  Theorem~\ref{thm:bernradetot}. Note that with this result, it is easier to compute the bound on the probability of overfitting (denoted as $\delta'$ in line 11: of Algorithm 1).\\

\noindent\textbf{Application of the Martingale Central Limit Theorem: }\label{sec:mclt}
In practical applications, one may prefer the standard practice in statistics of applying central limit asymptotic bounds. We develop here a bound based on the Martingale Central Limit Theorem (MCLT).  Our experimental results in Section~\ref{sec:exper} show that the bound obtained using the MCLT is more powerful while still preserving statistical validity. 

We adapt the following version of the MCLT~\footnote{Formally, the asymptotic is defined on a triangle array, where rows are samples of growing sizes. We also assume that all expectations are well-defined in the corresponding filtration.}:
\begin{theorem}[\emph{Corollary 3.2,} \cite{hall2014martingale}]
Let $Z_0,Z_1,\dots$ be a difference martingale with bounded absolute increments. 
Assume that (1)  
$\sum_{i=1}^n Z_i^2 \stackrel{p}\rightarrow V^2$ for a finite $V>0$, and (2) $\Exp{\max_i Z^2_i}\leq M <\infty$,
then $\sum_{i=1}^n Z_i/\sqrt{\sum_{i=1}^n \Exp{Z_i^2}}$ converges in distribution to $N(0,1)$.
\end{theorem}

When applying the MCLT, there is no advantage in bounding separately $\maxErr{\mathcal{F}_k, \bar{x}}-2R^{\mathcal{F}_k}_{m}$ and    $2R^{\mathcal{F}_k}_{m}-2\tilde{R}_{\bar{x},\ell}^{\mathcal{F}_k}$. Instead, we compute a bound on the distribution of $\maxErr{\mathcal{F}_k, \bar{x}}-2\tilde{R}_{\bar{x},\ell}^{\mathcal{F}_k}$ by analyzing the \emph{Doob supermartingale}
$$C_i=E[\maxErr{\mathcal{F}_k, \bar{x}}-2\tilde{R}_{\bar{x},\ell}^{\mathcal{F}_k}~|~Y_1,\dots,Y_i]$$
for $i=0,\dots,m(\ell+1)$,  with respect to the same $Y_1,\dots, Y_{m(\ell+1)}$ defined as in Section~\ref{sec:bound}.

As in the finite sample case, the following theorem relies on a careful analysis of $\Exp{Z_i^2}$ for the martingale difference sequence $Z_i =C_i-C_{i-1}$.
\begin{theorem}\label{thm:mclrad}Given a sample $\bar{x}\sim \mathcal{D}^m$, a family of functions $\mathcal{F}_k$ which take values in $[0,1]$,  $\ell$ independent vectors of Rademacher random variables, and $\epsilon, \delta \in (0,1)$, we have: 
\small
	\begin{equation*}
		\lim_{m\rightarrow \infty} \Prob{\maxErr{\mathcal{F}_k, \bar{x}}-2\tilde{R}_{\bar{x},\ell}^{\mathcal{F}_k}> \epsilon \frac{\sqrt{\ell +4\sqrt{\ell} +20}}{2 \sqrt{\ell m}}}< 1 - \Phi\left(\epsilon \right). 
	\end{equation*}
\normalsize
Where $\Phi(x)$ denotes the cumulative distribution function for the standard normal distribution.
\end{theorem}
\begin{IEEEproof}
The proof closely follows the steps of the proof of Theorem~\ref{thm:bern1bound}.
	Consider the Doob supermartingale for the function $\maxErr{\mathcal{F}_k,\bar{x}} - 2\tilde{R}^{\mathcal{F}_k}_{\bar{x},\ell}$:
	\begin{equation*}
		C_i = \Exp{\maxErr{\mathcal{F}_k,\bar{x}} - 2\tilde{R}^{\mathcal{F}_k}_{\bar{x},\ell}|Y_1,\ldots Y_i}~~ \mbox{for}~i=0,\dots,m(\ell+1),
	\end{equation*}
where for $1\leq i\leq m$,
$Y_i = X_i$, and the remaining $Y_i$ correspond to the $m\ell$ independent Rademacher random variables in the $\ell$ vectors. That is, $Y_{j(m)+i}= \sigma_{j,i}$, for $1\leq j\leq \ell$ and $1\leq i \leq m$.
Further, let us define the corresponding martingale difference sequence $Z_i=C_i-C_{i-1}$. 

In order to apply the MCLT, we need to bound $\sum_{i=1}^{m (\ell+1)}\Exp{Z_i^2}$ from above, and we need to verify that $|Z_i|$ is bounded.

Note that the sequence $Z_i$ defined here corresponds to the martingale difference sequence by the same name that we studied in the proof of Theorem~\ref{thm:bern1bound}. As shown in the proof of Theorem~\ref{thm:bern1bound}, we have $\sum_{i=1}^{m (\ell+1)}\Exp{Z_i^2} \leq \frac{\ell + 4\sqrt{\ell}+20}{4m\ell}$, and $|Z_i|\leq 2/m$ for all $1\leq i \leq m (\ell+1)$.
Applying the MCLT, we have that as $m$ goes to infinity,

\noindent $\sum_{i=1}^{m (\ell+1)} Z_i/\sqrt{\sum_{i=1}^{m (\ell+1)} \Exp{Z_i^2}}$ converges in distribution to $N(0,1)$, and thus:\small
\vspace{-1mm}
\begin{align*}
\lim_{m\rightarrow\infty} \Prob{\sum_{i=1}^{\ell (m+1)} Z_i \left(\sqrt{\sum_{i=1}^{m (\ell+1)} \Exp{Z_i^2}}\right)^{-1} > \epsilon} &< 1 -\Phi\left(\epsilon\right),\\
\lim_{m\rightarrow\infty} \Prob{\sum_{i=1}^{\ell (m+1)} Z_i > \epsilon\sqrt{\frac{\ell + 4\sqrt{\ell}+20}{4m\ell}}} &< 1 -\Phi\left(\epsilon\right),\\
\end{align*}\normalsize
\vspace{-2mm}
By linearity of expectation, and by applying Theorem~\ref{th:Ra}:
\begin{align*}
	\sum_{i=1}^{\ell (m+1)} Z_i &= \maxErr{\mathcal{F}_k,\bar{x}} - \Exp{\maxErr{\mathcal{F}_k,\bar{x}}} -2\left(\tilde{R}^{\mathcal{F}_k}_{\bar{x},\ell} - R^{\mathcal{F}_k}_{m}\right)\\
	&\geq \maxErr{\mathcal{F}_k,\bar{x}} - 2R^{\mathcal{F}_k}_{m} -2\left(\tilde{R}^{\mathcal{F}_k}_{\bar{x},\ell} - R^{\mathcal{F}_k}_{m}\right);\\
	&\geq \maxErr{\mathcal{F}_k,\bar{x}} - 2\tilde{R}^{\mathcal{F}_k}_{\bar{x},\ell};
\end{align*} 
The statement follows.
\end{IEEEproof}
Due to its asymptotic nature, it is not possible to compare directly the tightness of the bound in Theorem~\ref{thm:mclrad} with that of finite sample bounds such as the one in Theorem~\ref{thm:bernradetot}. Still, this bound is of great interest in many practical scenarios as it allows for a much tighter bound for the generalization error.
%

%% file: 04_algo.tex
\section{The RADABOUND Algorithm}\label{sec:stopping}
The algorithm starts by drawing  $\ell$ \emph{independent} vectors of Rademacher variables. These vectors are fixed throughout the execution of the algorithm. The advantage of fixing the Rademacher vectors is that (1) we deal with a nested sequence of events, $\mathcal{F}_{k-1}\subseteq \mathcal{F}_{k}$, and (2) the actual computation of the Rademacher complexity estimate is simple and efficient. \\

\noindent\textbf{Computing the estimate: } 
At the end of each round $k$, the algorithm stores for each of the Rademacher vectors $j=1,\dots, \ell$, the value
$\tilde{M}_{\bar{x},j}^{{\mathcal F}_{k}} =\max_{f\in \mathcal{ F}_k}
\frac{1}{|\bar{x}|}\sum_{i=1}^{|\bar{x}|} f(x_i)\mathbf{\sigma}_{i,j}$. 
To update these values, at iteration $k+1$, the algorithm computes 
$$\tilde{M}_{\bar{x},j}^{{\mathcal F}_{k+1}} \leftarrow \max \{\tilde{M}_{\bar{x},j}^{{\mathcal F}_{k}} , \frac{1}{|\bar{x}|}\sum_{i=1}^m f_{k+1}(x_i)\mathbf{\sigma}_{i,j}\}, ~~ j=1,\dots, \ell.$$
The estimate of the \Rade{} at round $k+1$ is then given by $\tilde{R}_{\bar{x}}^{{\mathcal F}_{k+1}}= \frac{1}{\ell} \sum_{j=1}^\ell \tilde{M}_{\bar{x},j}^{{\mathcal F}_{k+1}}$.


\begin{algorithm}[ht!]
\caption{RADABOUND - Adaptive data analysis with Rademacher Complexity control}\label{alg:RADABOUND}
\begin{algorithmic}[1]
\Procedure{RADABOUND}{$\bar{x},\varepsilon,\delta,\ell$}
\State $m\leftarrow \abs{\bar{x}}$ \Comment{Size of the input sample}
\Statex{$\triangleright$ Initialization estimator for Rademacher Complexity \hfill}
\For{$j\in\{0,1,\ldots,\ell\}$} 
\State $\sigma_j \leftarrow $ vector of $m$ iid Rademacher RVs
\State $R^{\mathcal{F}_0}_{\bar{x},j}\leftarrow 0$
\EndFor 
\Statex{$\triangleright$ Main execution body}
\While{\text{new }k\text{-th  query} $f_k$ from the stream} 
	\State $\mathcal{F}_{k+1} \leftarrow \mathcal{F}_{k}\cup \{f_k\}$
	\Statex{\qquad\quad $\triangleright$ Rademacher Average estimation update}
	\For{$j\in\{0,1,\ldots,\ell\}$} 
		\State $R_{\bar{x},j}^{{\mathcal F}_{k+1}} \leftarrow \max \{R_{\bar{x},j}^{{\mathcal F}_{k}} , \frac{1}{m} \sum_{i=1}^m f_k(x_i)\sigma_{j,i} \}$
	\EndFor
	\State $\tilde{R}_{\bar{x},\ell}^{{\mathcal F}_{k+1}} \leftarrow \frac{1}{\ell}\sum_{j=1}^\ell R_{\bar{x},j}^{{\mathcal F}_{k+1}}$
    \Statex{\qquad\quad $\triangleright$ Control with BIM}
	\State
	$\delta'\leftarrow e^{-\left( \min {0, \epsilon-2\tilde{R}^{\mathcal{F}_k}_{\bar{x},\ell}}\right)^2/\frac{\ell + 4\sqrt{\ell}+20}{2m\ell}+\frac{4 \epsilon}{3m}}$
   \Statex{\qquad\quad $\triangleright$ Control with MCLT- Alternative to 11:}
    \State \begin{varwidth}[t]{\linewidth}{\bf or} $\delta'\leftarrow 1-\Phi\left( \max\{0, \epsilon-2\tilde{R}_{\bar{x},\ell}^{{\mathcal F}_{k+1}}\}\sqrt{\frac{4\ell m}{\ell + 4\sqrt{\ell}+20}} \right)$  \end{varwidth}
   \Statex{\qquad\quad $\triangleright$ Overfit control test}
    \If{$\delta' \leq \delta(1-\delta)$} 
    	\State{\textbf{return} $\frac{1}{m} \sum_{x \in \bar{x}} f(x)$}
    \Else
    	\State \begin{varwidth}[t]{\linewidth}\textbf{Halt: Cannot guarantee the statistical\\ validity of further queries.}\end{varwidth}
    \EndIf
\EndWhile
\EndProcedure
\end{algorithmic}
\end{algorithm}
\noindent\textbf{Stopping rule: }
Given real values $\epsilon,\delta \in (0,1)$, the procedure halts at the first ${k}$-th step for which it cannot guarantee that $$\Probd{\mathcal{L}}{\Psi (\mathcal{ F}_{k+1},\bar{x}) > \epsilon} \leq \delta.$$ 
Recall from~\eqref{eq:iterativeproc} that 
\begin{equation}\label{eq:condstop2}
	\Probd{\mathcal{L}}{\Psi (\mathcal{ F}_k,\bar{x})> \epsilon}
\leq \frac{\Prob{\Psi (\mathcal{ F}_k,\bar{x})> \epsilon}}{\Prob{\Psi (\mathcal{ F}_{k-1},\bar{x})\leq \epsilon}}.
\end{equation}
Since $\Pr(\Psi (\mathcal{ F}_k,\bar{x})> \epsilon)\geq \Pr(\Psi (\mathcal{ F}_{k-1},\bar{x})> \epsilon)$, it is sufficient to require $\Prob{\Psi (\mathcal{ F}_k,\bar{x})>\epsilon}< \delta (1-\delta)$ to have $\Probd{\mathcal{L}}{\Psi (\mathcal{ F}_k,\bar{x})> \epsilon}\leq \delta,$ and we can use the bounds obtained in Theorem~\ref{thm:bernradetot} or Theorem~\ref{thm:mclrad}. Thus, we prove

\begin{theorem}
Given a sample $\bar{x}\sim\mathcal{D}^m$, let $\mathcal{F}_k$ denote the set of functions adaptively selected during the first $k$ steps. If \algo{} has not halted at step $k$,  then 
$$
\Probd{\mathcal{L}}{\Psi (\mathcal{ F}_k,\bar{x})\leq \epsilon} >1- \delta.
$$	
\end{theorem}
The bound in Theorem~\ref{thm:mclrad} based on the MCLT can be used in \algo{} \emph{as an alternative} to the bound in Theorem~\ref{thm:bernradetot} (lines 11-12 in Algorithm 1). In Section 5, we present an experimental comparison of performance of \algo{} when using the two methods.

%% file: 05_experiments.tex
\begin{figure*}[ht!]
\makebox[\textwidth][c]{
 \begin{minipage}{0.35\textwidth}
\captionsetup{width=.85\linewidth}
\begin{subfigure}{\textwidth}
\centering
\includegraphics[width=\textwidth]{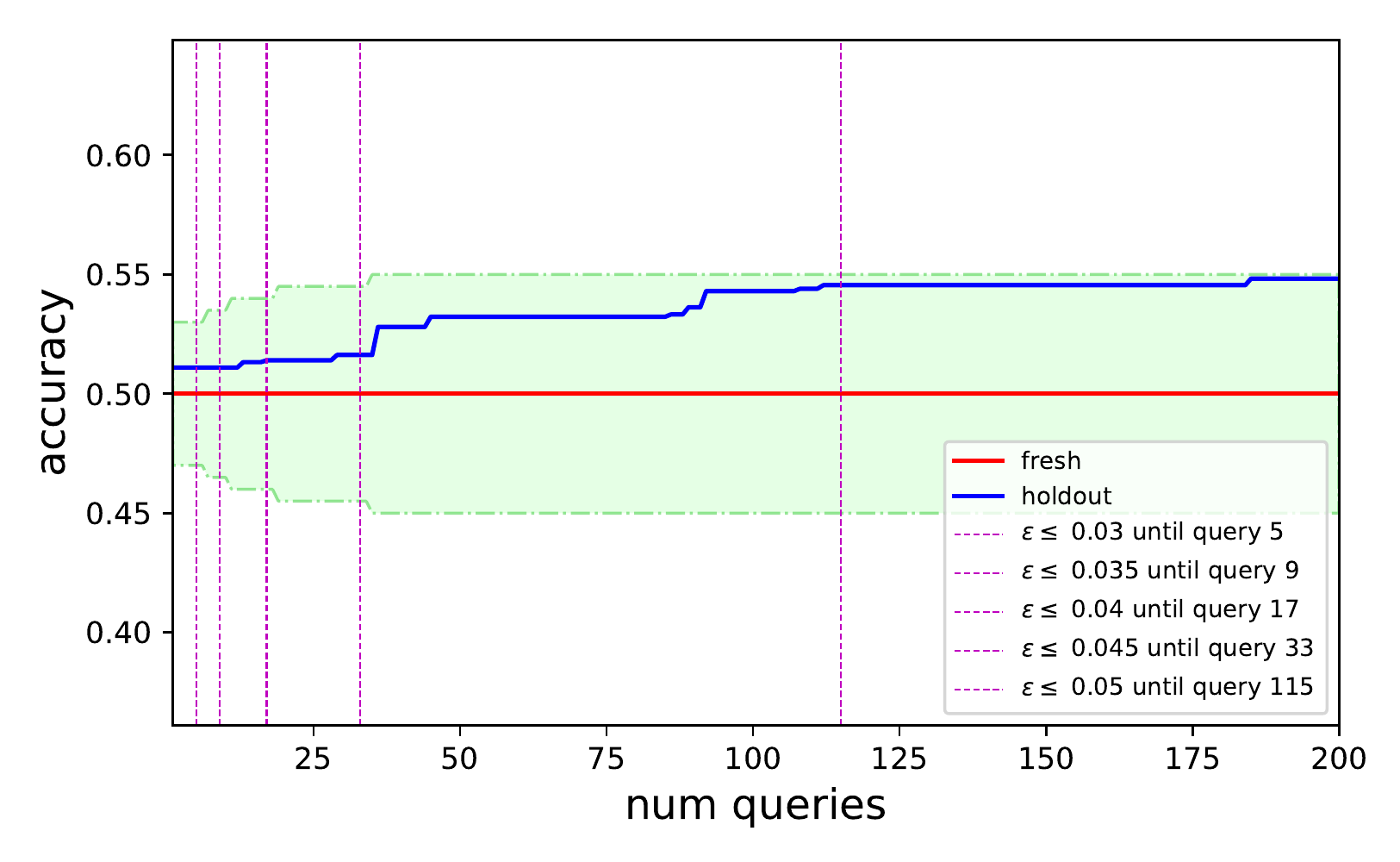}
\caption{Bernstein bound}
\label{fig:exp1}
\end{subfigure}\hfill
\begin{subfigure}{\textwidth}
\centering
\includegraphics[width=\textwidth]{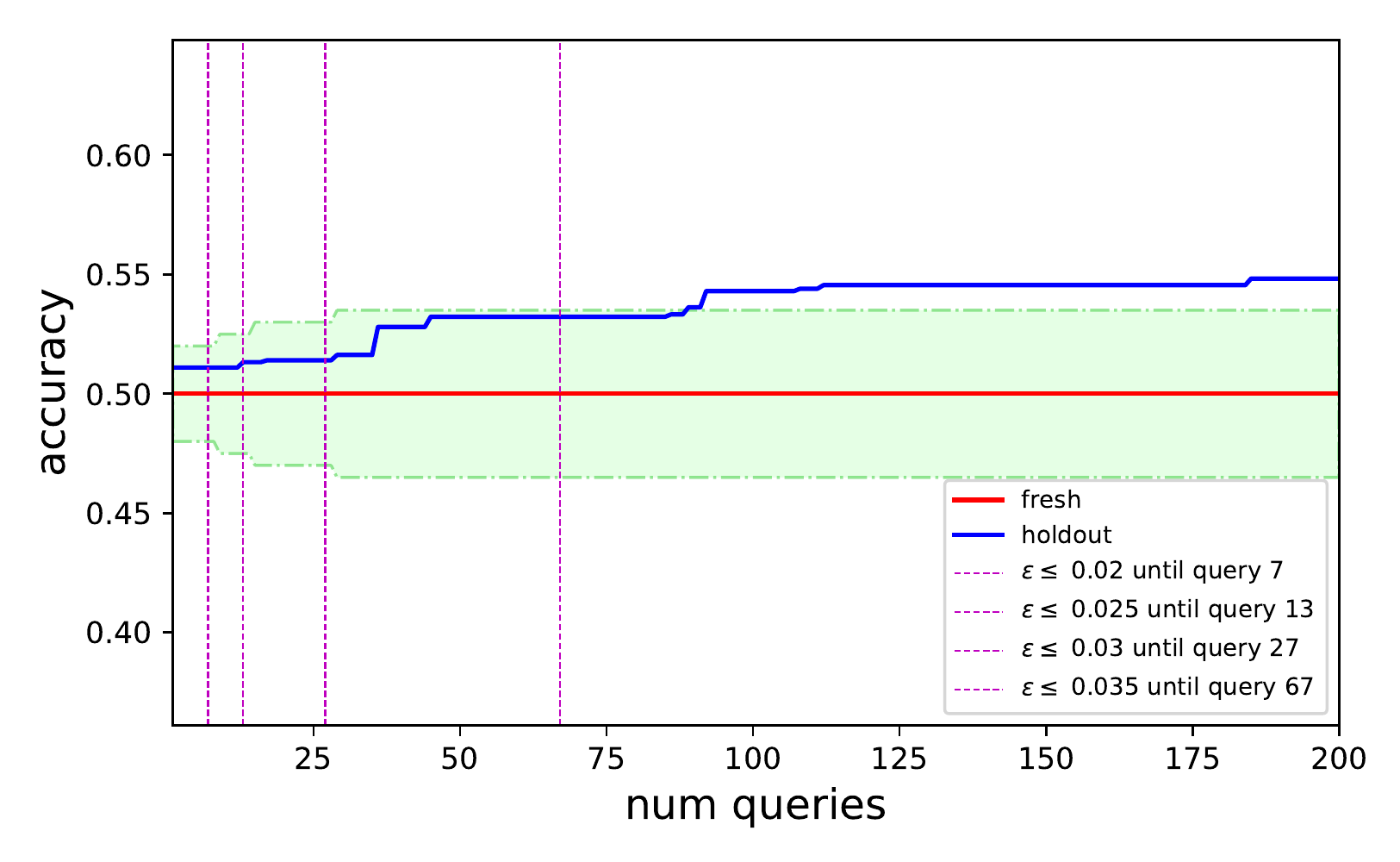}
\caption{MCLT}
\label{fig:exp2}
\end{subfigure}\hfill
\caption{No signal. Feature values from $\mathcal{N}(0,1)$. $\delta = 0.1$.}
   \end{minipage}\hfill
\begin{minipage}{0.35\textwidth}
\captionsetup{width=.85\linewidth}
     \centering
     \begin{subfigure}{\textwidth}
\centering
\includegraphics[width=\textwidth]{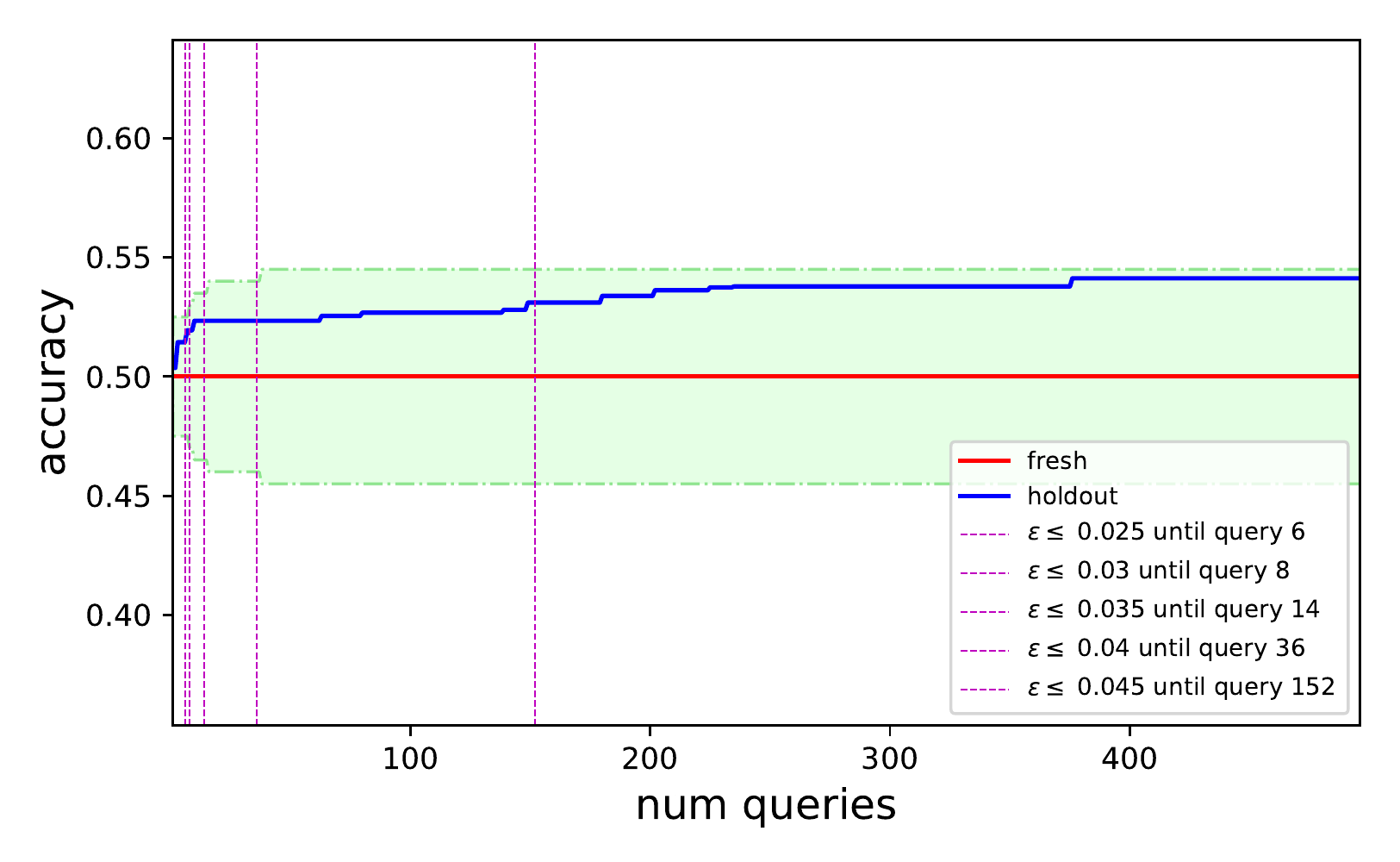}
\caption{Bernstein bound}
\label{fig:exp1}
\end{subfigure}\hfill
\begin{subfigure}{\textwidth}
\centering
\includegraphics[width=\textwidth]{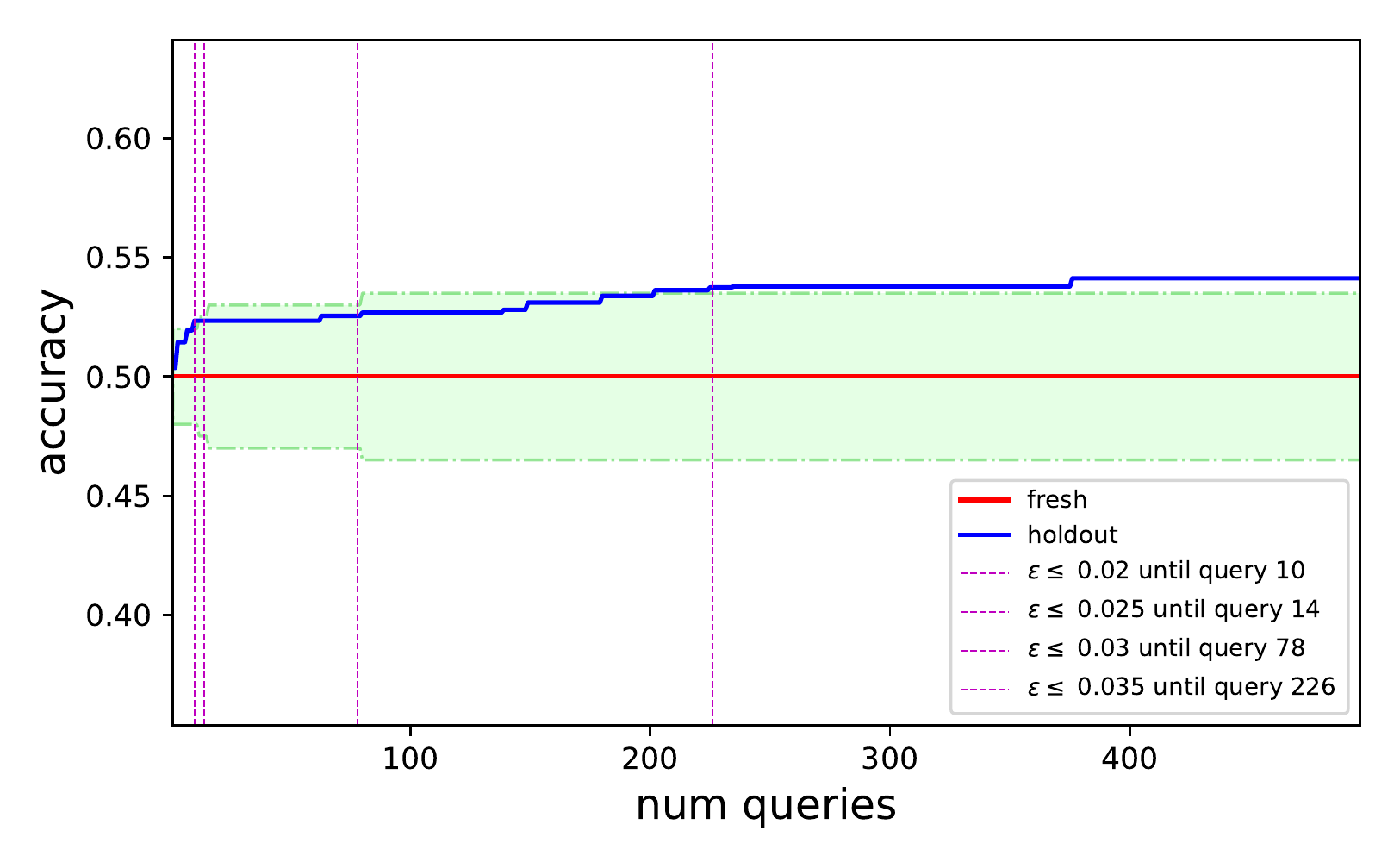}
\caption{MCLT bound}
\label{fig:exp2}
\end{subfigure}\hfill
\caption{No signal. Feature values from $\mathcal{N}(0,2)$. $\delta = 0.15$.}
   \end{minipage}\hfill
\begin{minipage}{0.35\textwidth}
\captionsetup{width=.85\linewidth}
     \centering
     \begin{subfigure}{\textwidth}
\centering
\includegraphics[width=\textwidth]{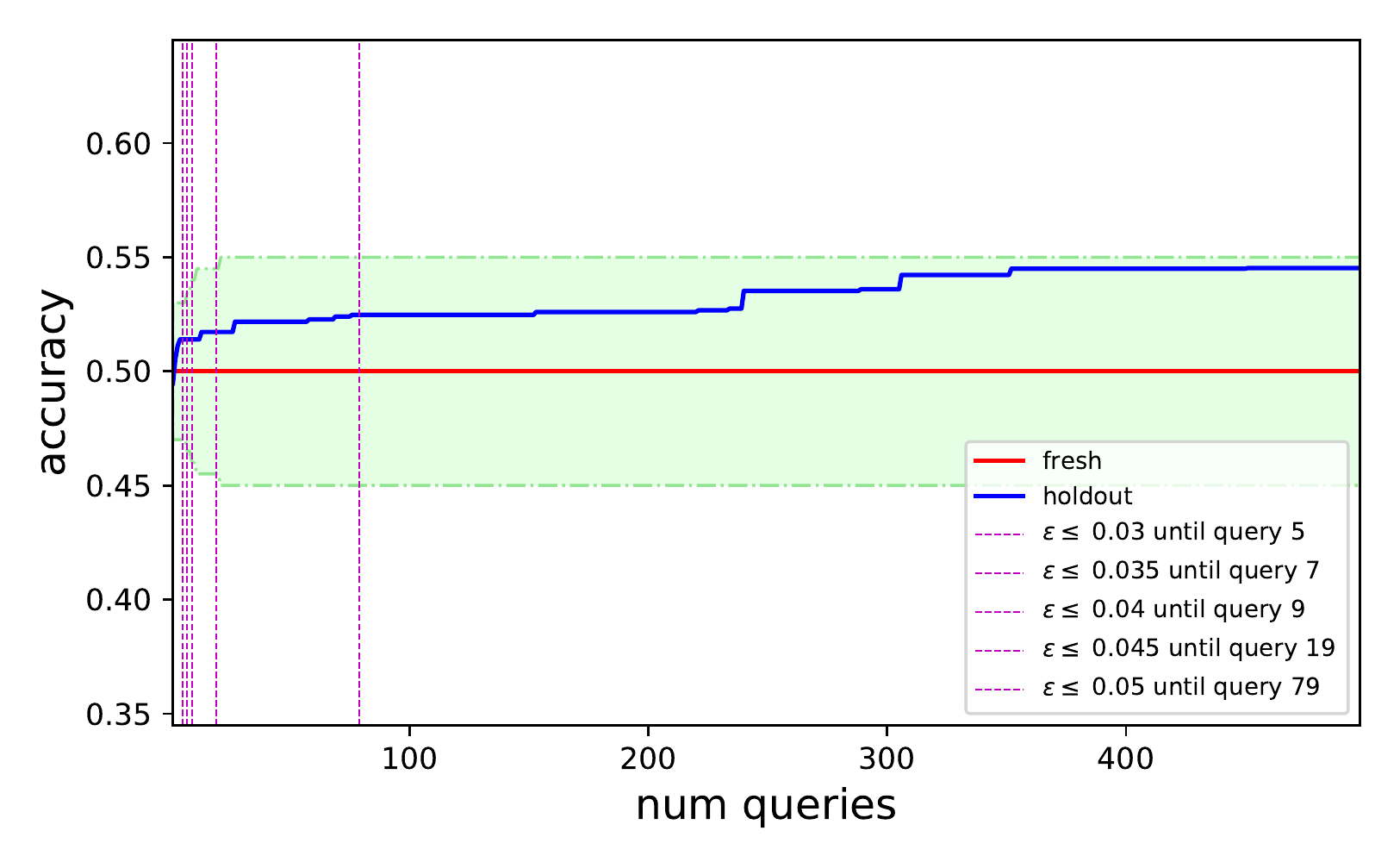}
\caption{Bernstein bound}
\label{fig:exp1}
\end{subfigure}\hfill
\begin{subfigure}{\textwidth}
\centering
\includegraphics[width=\textwidth]{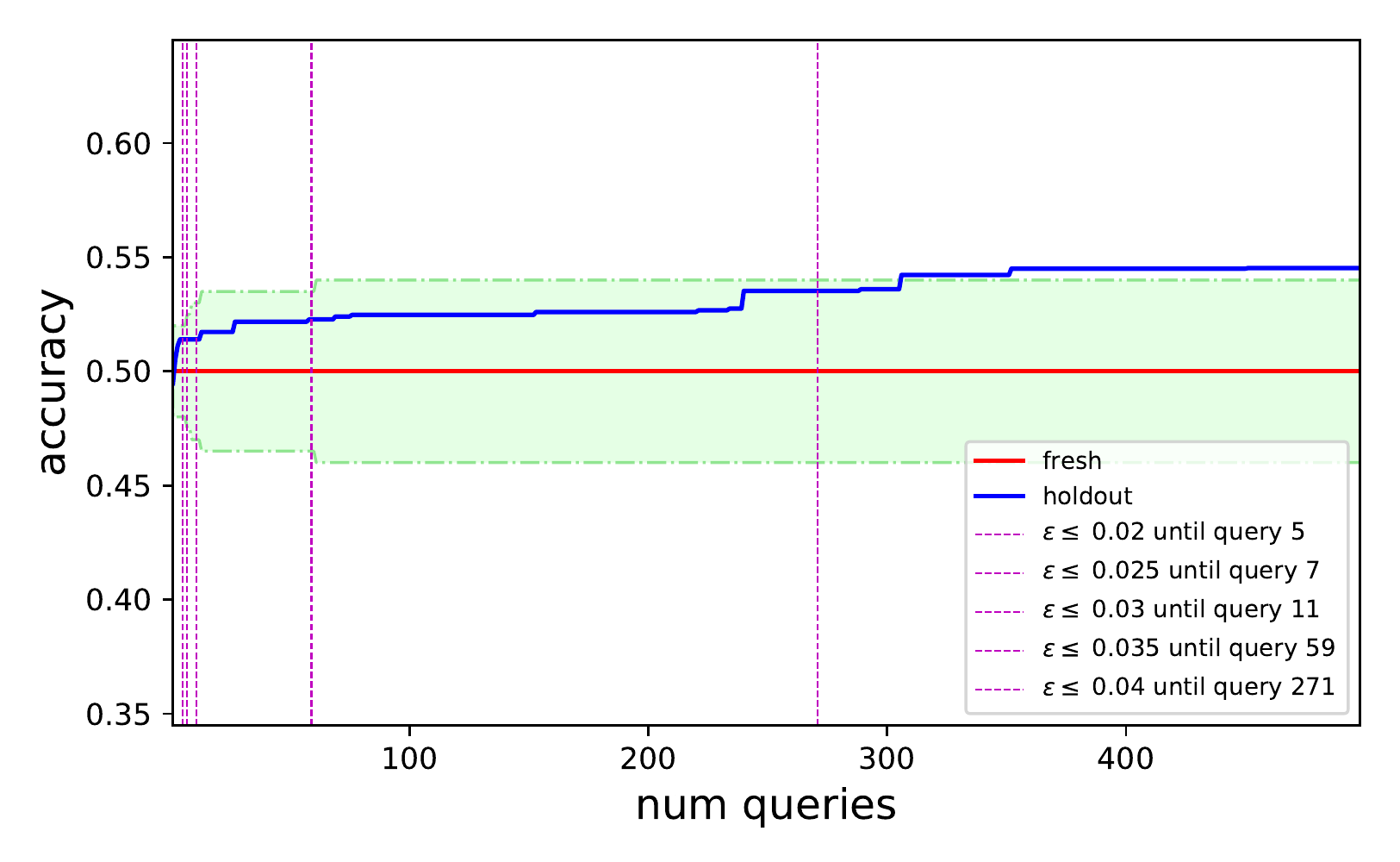}
\caption{MCLT bound}
\label{fig:exp2}
\end{subfigure}\hfill
\caption{No signal. Feature values from $\mathcal{N}(0,8)$. $\delta = 0.2$.}
   \end{minipage}
   }
   \vspace{-3mm}
\end{figure*}
\section{Experimental results}\label{sec:exper}
We demonstrate the power and efficiency of our technique through a variety of experiments. Our experimental setup is similar to the one used in the state-of-the-art~\cite{Dwork15generalization}, except that all our reported results are for ranges of parameters for which we actually have provable statistical guarantees.\\
\vspace{-2mm}

\noindent\textbf{Experimental Setup: }We consider a learning task of classifying vectors composed by $d$ features to the classes ``-1'' or ``1''.
We consider only linear classifier vectors $\textbf{w}\in \{-1,0,1\}^d$,  assigning vector $\textbf{x}$ to class $\textbf{h}(\textbf{x})= \textnormal{sign}\left(\textbf{w}\cdot\textbf{x}\right)$.
The goal of the learning algorithm is to find a classifier with minimum expected loss for the $0,1$ \emph{hard loss function} (0 for correct classification, 1 otherwise). To model a typical learning scenario, the learning algorithm is given two independent datasets.
A \emph{training set} $X_T$ and an \emph{holdout set} $X_H$. We then evaluate the performance of the learning  algorithm using a third, independent \emph{fresh set}. Evaluation on the fresh, independent set provides a baseline for the actual performance of the classifier being considered on an independent dataset.
 The goal of the algorithm being tested if to obtain a classifier $h()$ which ``\emph{fits the data}'' as best as possible, that is the a classifier which minimizes the \emph{Empirical Risk} over the holdout data, that is $\frac{1}{m} \sum_{\textbf{x}\in X_H} \ell\left(h(\textbf{x}), y(\textbf{x})\right)$, where $y(\textbf{x})$ denotes the ``\emph{true label}'' associated with the data point $\textbf{x}$ and $\ell\left(h(\textbf{x}), y(\textbf{x})\right)$ denotes the loss function. In our setting we consider the $\{0,1\}$ loss function such that $\ell\left(h(\textbf{x}), y(\textbf{x})\right)= 0$ if $h(\textbf{x}), y(\textbf{x}$ (i.e., the classifier $h()$ assigns the correct label to \textbf{x}), or  $\ell\left(h(\textbf{x}), y(\textbf{x})\right)= 1$ otherwise.


Our learning algorithm works as follows:
In the first phase, the algorithm  evaluates the correlation between the values of the features and the labels of the vectors using \emph{only} the training dataset $X_T$ as $c_i = \frac{1}{m} \sum_{\textbf{x}\in X_T} \textbf{x}[i]l(\textbf{x})$. The features are then sorted (in descending order) according to the absolute values of their correlations $|c_i|$ to the labels. 

The actual adaptive analysis of the data occurs in the second phase of the algorithm using the holdout data $X_H$. 
The algorithm starts with a classifier $\textbf{w}=0$. 
It then considers features according to the order computed in the  first phase. 
That is, features with \emph{stronger} correlation (either positive or negative) to the label are considered first. 
Using the holdout set $X_H$,
the algorithm tests, for each feature, whether assigning weight -1 or 1 to it improves the performance of the current best classifier.
If that is the case, the classifier is updated with the new value for the feature;
otherwise, the feature is left with weight zero. Each newly tested classifier is added to the class function $\mathcal{F}_k$. \algo{} then computes a new estimate $\tilde{R}_m^{\mathcal{F}_{k}}$ of the \Rade{} of $\mathcal{F}_k$, and uses it to determinate whether the total accumulated error is below $\epsilon$ with probability at least $1-\delta$ as discussed in Section~\ref{sec:stopping}.



\noindent\textbf{Data generation: }In all the experiments $|X_T|=|X_H|= 4000$. Each vector in the dataset has 500 features. The estimation of the \Rade{} $\tilde{R}_{\bar{x},\ell}^{{\mathcal F}}$ is computed according to~\eqref{eq:radestimate} using $\ell =32$ vectors of Rademacher random variables. We report results using (a) Bernstein's Inequality (Section~\ref{sec:bound}) and, (b) the MCLT (Section~\ref{sec:mclt}).
%
We consider the two following scenarios: 
\begin{itemize}
    \item {\textbf{No signal in the data:}}
In this setting, each point $\textbf{x}\in X_H$ is assigned a label 
independently and uniformly at random. The feature values are taken \emph{independently} from a normal distribution with expectation 0 and various variance values. Thus, there is no correlation between the labels and the values of the features. We report the results in Figures 1-3.

\item {\textbf{Signal in the data:}} In this setting 
the ``\emph{strength}'' of the correlation between some features and the labels is characterized by two parameters: $n$, the number of the queries whose value is correlated to the label, and (positive or negative) $bias$ which defines the strength and sign of the correlation. We first generate datasets with no signal, like in the previous setting. We then fix a set of $50$  features to be  correlated with the label of their vectors. Letting $l(\textbf{x})$ be the label of vector $\textbf{x}$, the $50$ correlated features of $\textbf{x}$ are modified by adding $bias\times l(\textbf{x})$ to their original value. We report the results in Figures 4-6.
\end{itemize}
\begin{figure*}[ht]
\makebox[\textwidth][c]{
   \begin{minipage}{0.335\textwidth}
   \captionsetup{width=.85\linewidth}
     \centering
     \begin{subfigure}{\textwidth}
\centering
\includegraphics[width=\textwidth]{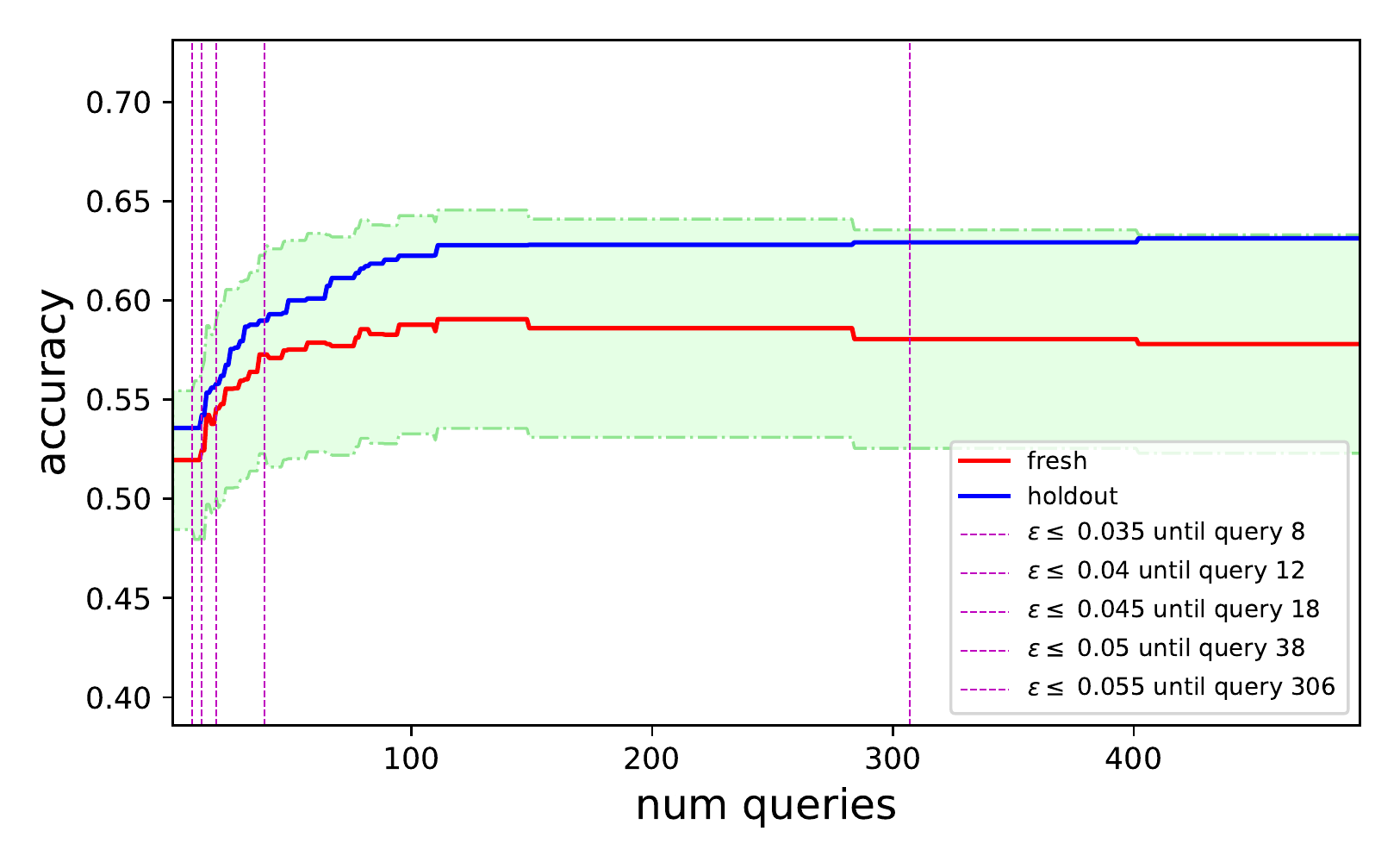}
\caption{Bernstein bound}
\label{fig:exp1}
\end{subfigure}\hfill
\begin{subfigure}{\textwidth}
\centering
\includegraphics[width=\textwidth]{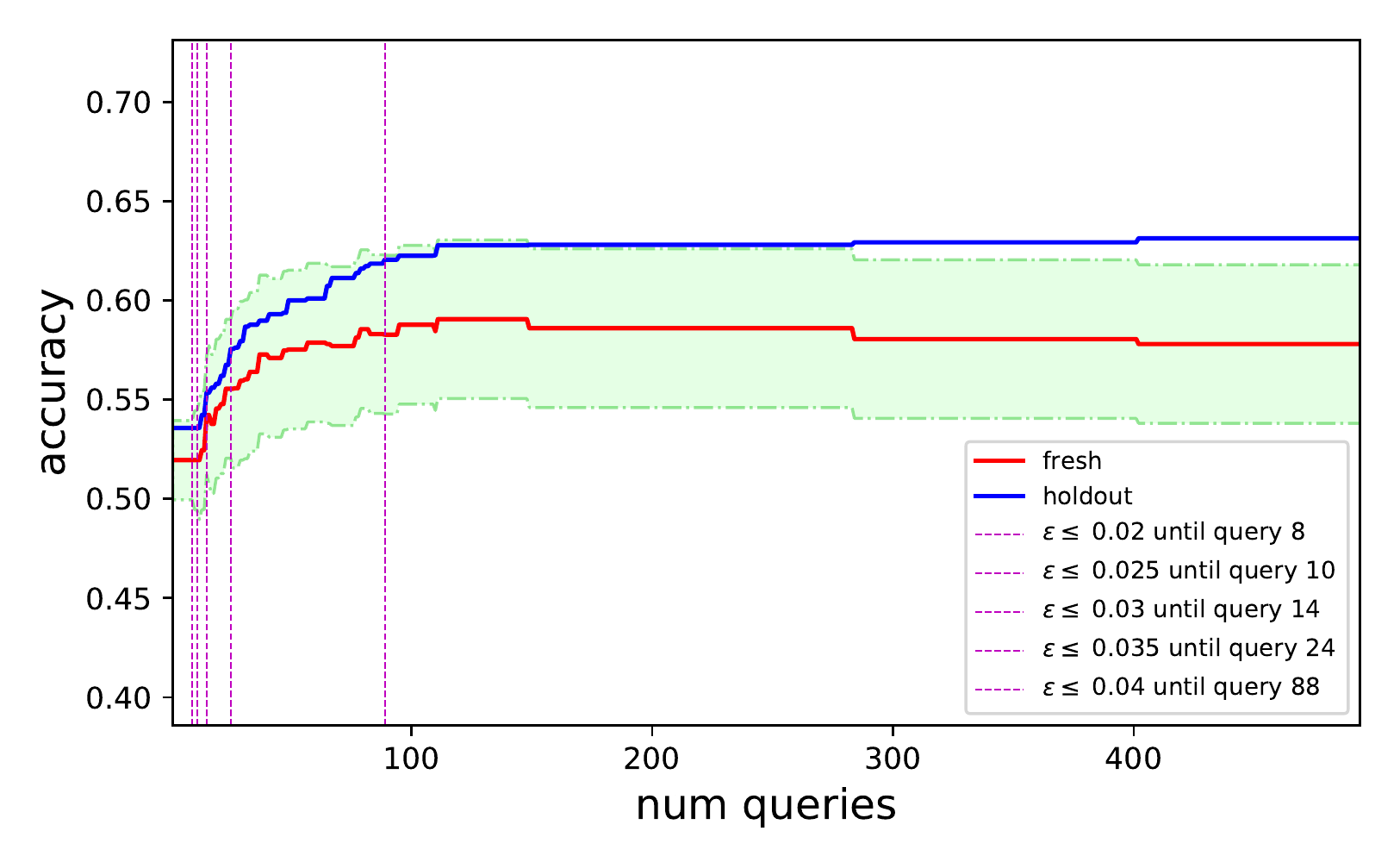}
\caption{MCLT bound}
\label{fig:exp2}
\end{subfigure}\hfill
\caption{Signal. Feature values from $\mathcal{N}(0,4)$, $\delta = 0.1$, $bias = 0.5$.}
   \end{minipage}\hfill
 \begin {minipage}{0.335\textwidth}
 \captionsetup{width=.85\linewidth}
     \centering
     \begin{subfigure}{\textwidth}
\centering
\includegraphics[width=\textwidth]{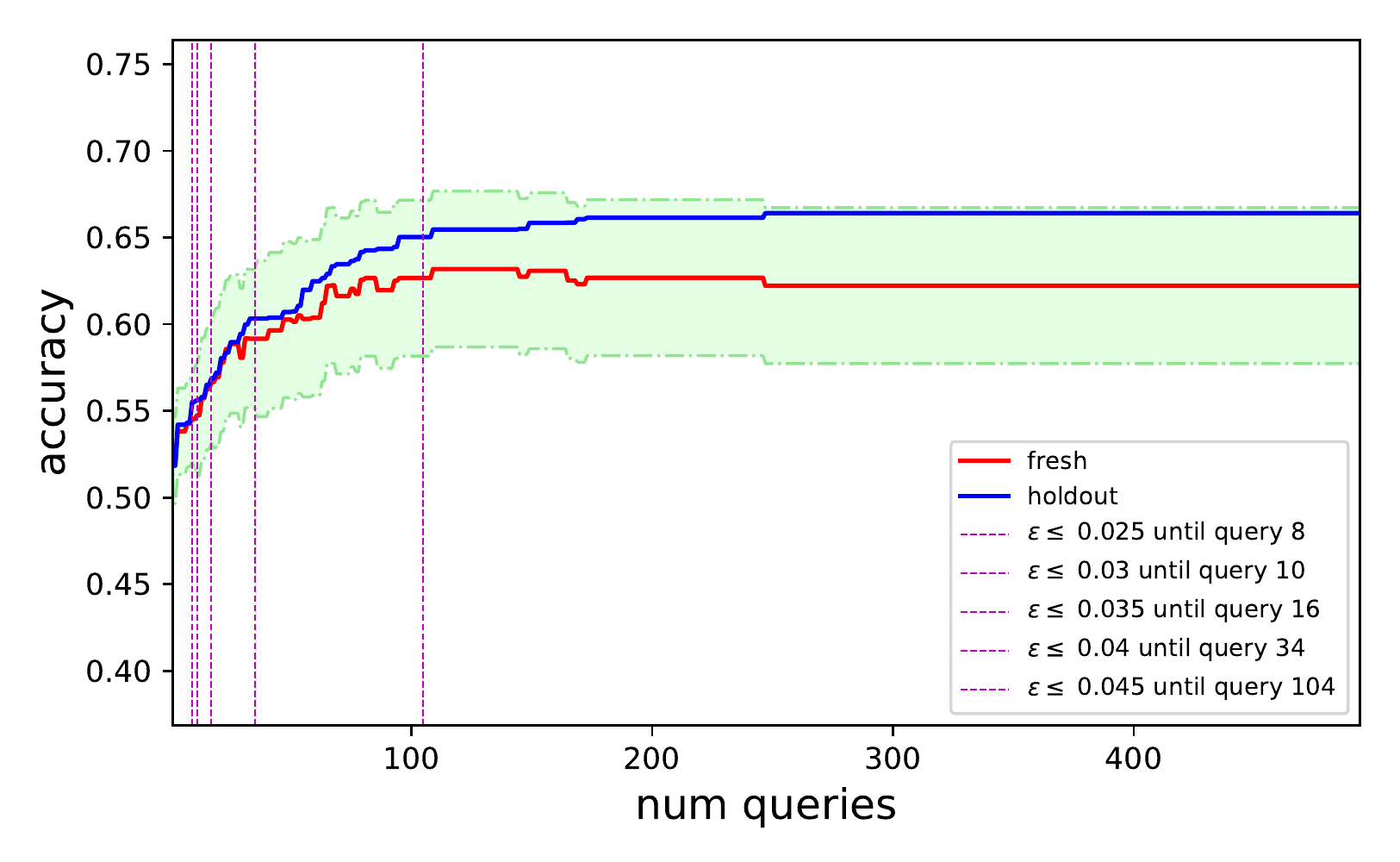}
\caption{Bernstein bound}
\label{fig:exp1}
\end{subfigure}\hfill
\begin{subfigure}{\textwidth}
\centering
\includegraphics[width=\textwidth]{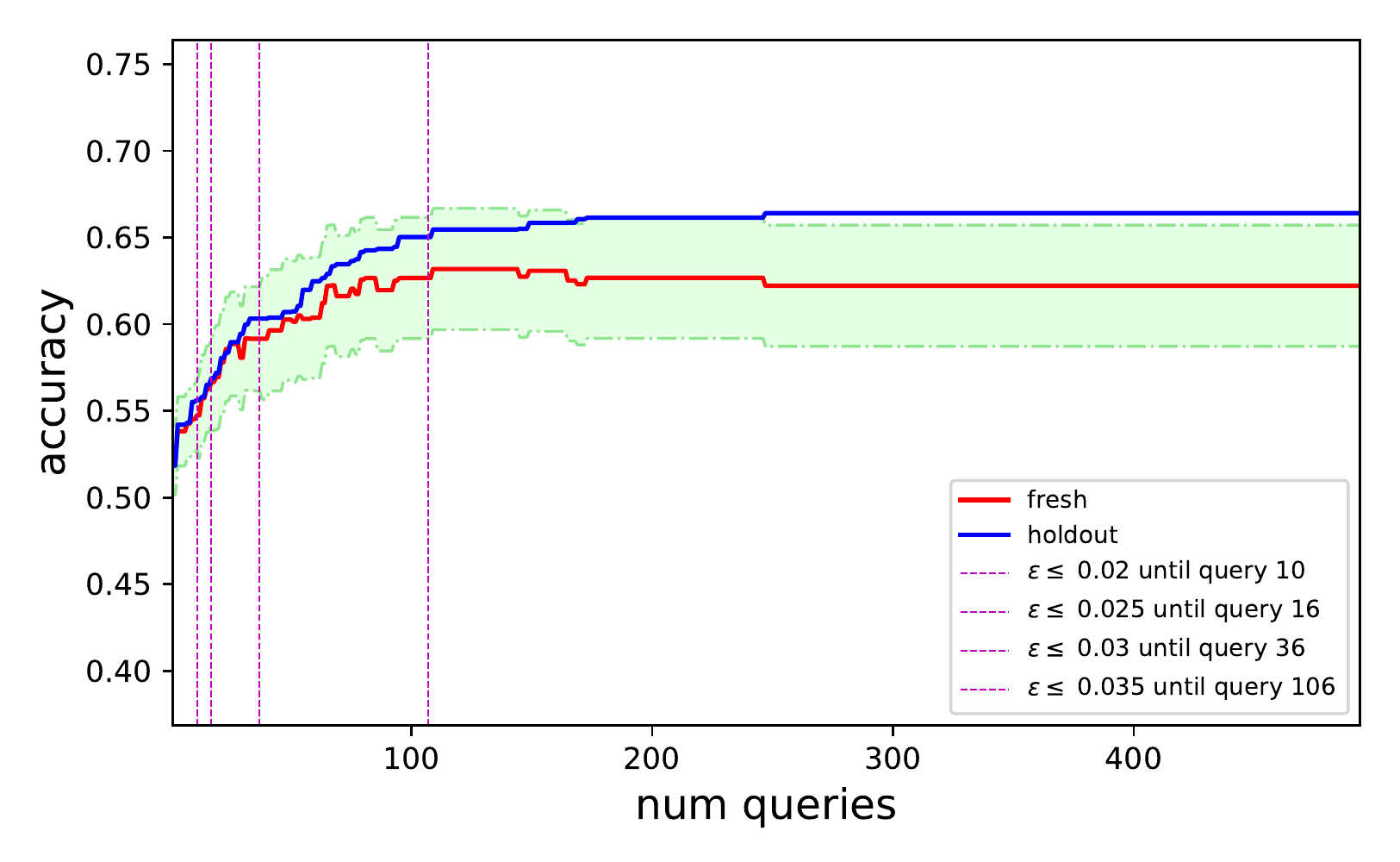}
\caption{MCLT bound}
\label{fig:exp2}
\end{subfigure}\hfill
\caption{Signal. Feature values from $\mathcal{N}(0,2)$, $\delta = 0.15$, $bias = 0.5$.}
   \end{minipage}\hfill
   \begin {minipage}{0.335\textwidth}
   \captionsetup{width=.85\linewidth}
     \centering
     \begin{subfigure}{\textwidth}
\centering
\includegraphics[width=\textwidth]{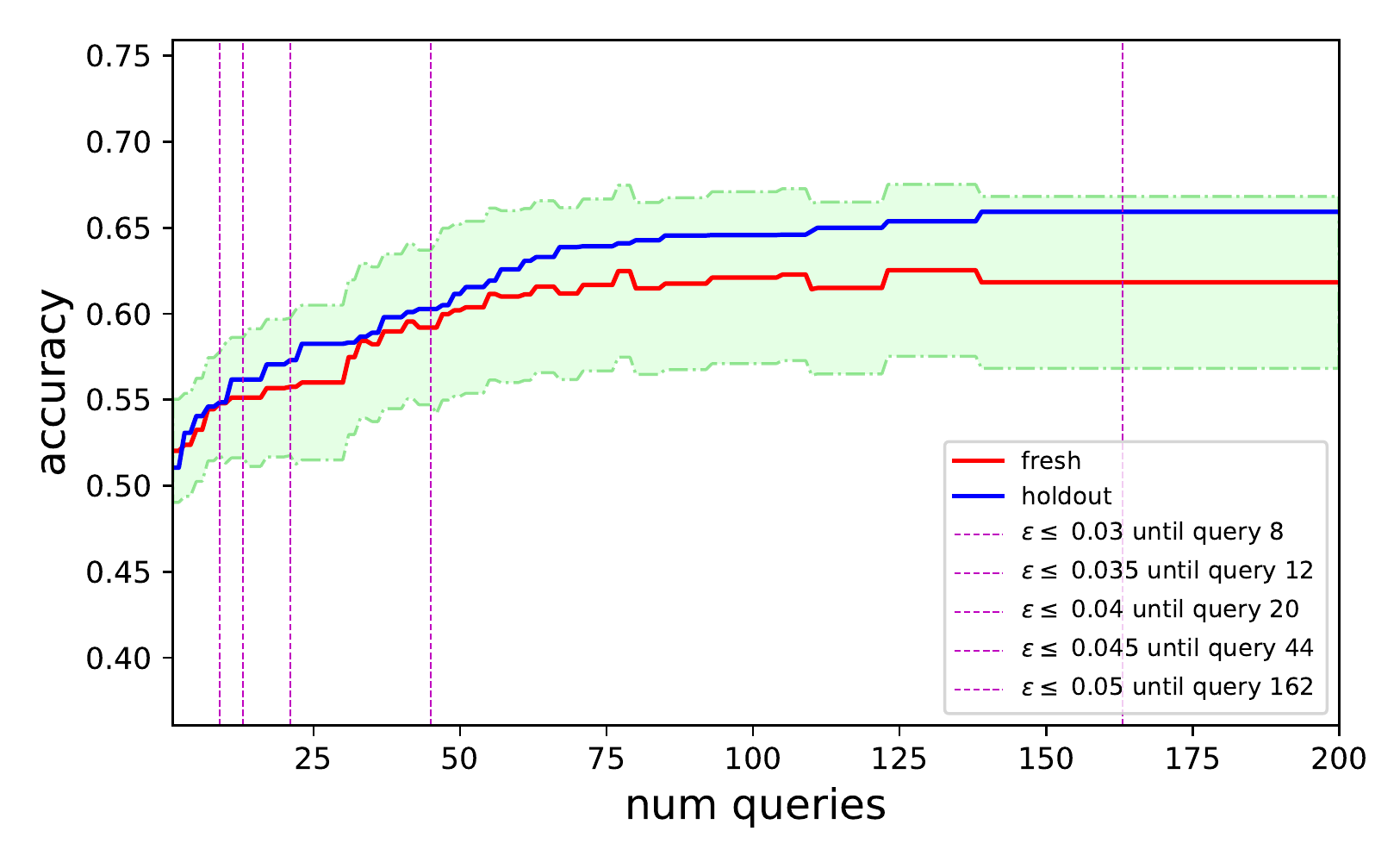}
\caption{Bernstein bound}
\label{fig:exp1}
\end{subfigure}\hfill
\begin{subfigure}{\textwidth}
\centering
\includegraphics[width=\textwidth]{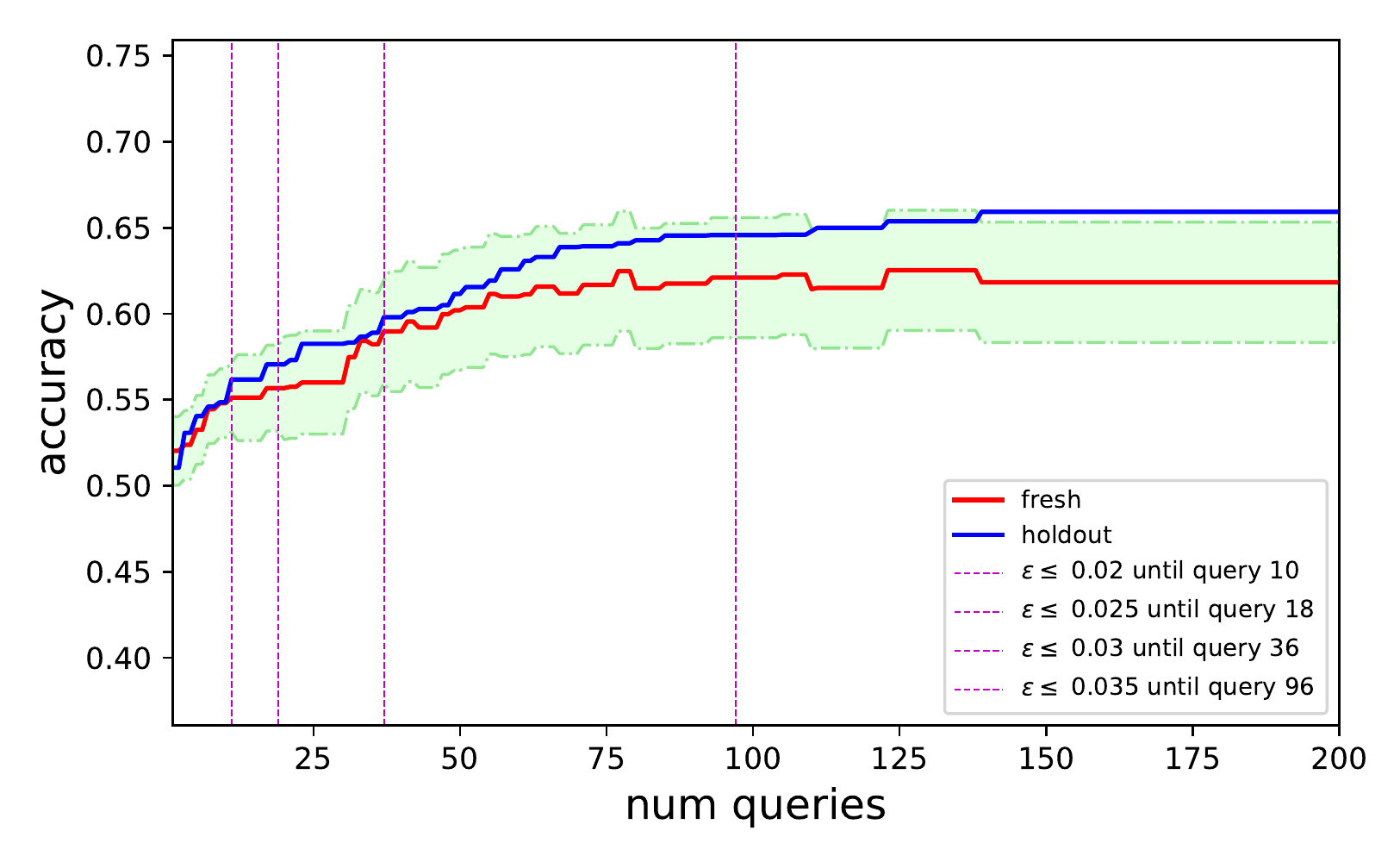}
\caption{MCLT bound}
\label{fig:exp2}
\end{subfigure}\hfill
\caption{Signal. Feature values from $\mathcal{N}(0,2)$, $\delta = 0.2$, $bias = 0.25$.}
   \end{minipage}
   }
   \vspace{-2mm}
\end{figure*}
Our experimental evaluation uses \emph{synthetic data} generated randomly according to the previously discussed specifications. Besides being widespread in statistics literature and used in our most direct term of comparison~\cite{dwork2015reusable}, using synthetic data is particularly useful in this setting for two main reasons: (i) it allows evaluate the performance of our testing procedure while evaluating different settings by modifying its parameters (e.g., the size of the input sample, the strength of the signal being observed), and (ii), most importantly, it allows to verify whether/when overfitting (due to adaptivity) has \emph{actually occurred} (i.e., when the blue line in the figures exceeds the green shaded region). This allows to asses both the \emph{correctness} of the method and its \emph{power}.
%


\noindent \textbf{Results: } The results of the experiments in the no signal (resp., signal in the data) setting are displayed in Figures 1-3 (resp., Figures 4-6). Each figure above corresponds to several runs of an experiment with the same parameters but different values of $\epsilon$ (the error bound). The blue line gives the accuracy of the best classifier computed after running the corresponding number of queries on the holdout set $X_H$. The red line gives the accuracy of the same classifier on fresh data (i.e., the red line represents the ground truth performance of the classifier being evaluated on a fresh sample). The vertical bars give the computed stopping time for each value of $\epsilon$ according to the ``\emph{stopping rule}'' of \algo{}. The shaded green area corresponds to $\pm \epsilon$ values around the ``\emph{true accuracy}'' of the classifier (the red line), for the $\epsilon$ value associated with the next vertical bar. The green shaded area beyond the last bar uses the same $\epsilon$ as the last bar. 
In a correct execution of \algo{}, the blue line does not exit the green shaded area before the last vertical bar.
The power of \algo{} is measured by how close is the last bar to the first time the green line exit the shaded area. 

The results of the experimental evaluations, as presented in the figures, demonstrate that \algo{} successfully halts the sequence of tests before overfitting for the various values of $\epsilon$, as the green line corresponding to the values of the function evaluated on the holdout does not exit the green shaded area before the corresponding vertical bar. The statistical power of the procedure is highlighted in particular by the result of experiments for which there is an actual correlation between the labels (Figures 4-6) and the value of the features as the overfitting control ensured by \algo{} is not achieved at the expense of detecting the  signal in the data.

In several scenarios (e.g., Figures 4 and 6), \algo{} halts its execution very close to the first iteration for which overfit (with respect to the value of $\epsilon$) actually occurs. \algo{} does not appear to be influenced by the distribution $\mathcal{D}$ over the data, but rather it behaves differently depending on the actual family of functions being tested. 

For similar $\epsilon, \delta$ parameters, the state-of-the-art $\mathtt{Thresholdout}$ algorithm~\cite{Dwork15generalization} would require an holdout dataset of size $\sim 4\times 10^6$ to provide answers to just 10 queries (details in Section~\ref{sec:dwork}). 
In contrast, our experiments show that \algo{} can provably handle such parameters with a holdout set of just 4000 samples, thus with an improvement of almost three orders of magnitude in terms of \emph{sample complexity}. The comparison is further discussed in Section~\ref{sec:dwork}. Using the MCLT leads to a tight analysis of $\maxErr{\mathcal{ F}_k,\bar{x}}$, which in turn allows testing a higher number of adaptively chosen classifiers before halting and without overfitting (Figures 1-6b), compared to the stopping points obtained using the BIM (Figures 1-6a).

Finally, the fact that even when using the bounds obtained  using the MCLT, the procedure halts correctly, further suggests that, despite their ``asymptotic'' nature, these are actually highly reliable even when dealing with an input sample of relatively small dimension.  

%% file: 07_diffprivacy.tex
\section{Comparison with methods based on Differential Privacy}
\label{sec:dwork}

The $\mathtt{Thresholdout}$ algorithm~\cite{Dwork15generalization} provides guarantees similar to those of \algo{}. $\mathtt{Thresholdout}$ operates using two datasets: a public dataset and a private \emph{holdout} dataset. Every time a new query is received the algorithm evaluates its value on both the public and the private dataset. If their absolute difference is within a given \emph{threshold}, $\mathtt{Thresholdout}$ returns to the user the value observed on the public dataset after perturbing it with some noise. Viceversa, if the absolute difference is higher than a certain threshold, the algorithm detects that the query being considered is overfitting on the public dataset. In this case, $\mathtt{Thresholdout}$ may instead provide the value computed on the private dataset after perturbing it with noise. As this last operation effectively ``\emph{leaks}'' information regarding the holdout it can be executed up to $B$ times, where $B$ must be fixed prior to the execution of the algorithm. 

To characterize the number of queries for which the $\mathtt{Thresholdout}$ algorithm provides provable statistical guarantees we apply Theorem 25 in~\cite{Dwork15generalization}~\footnote{A careful reader will notice that  Figures 1-3 in ~\cite{Dwork15generalization} (the same figure appears as Figures 1 and 2 in~\cite{dwork2015reusable}) represent an idealized illustration rather than statistically valid results. For the sample size used in these figures, the bound on the error probability of the threshold algorithm is not smaller than 1. Furthermore, the results crucially depend on a preprocessing of the two data sets (lines 95-97 in the ``$\mathtt{runClassifier}$'' procedure in the python code in the supporting materials) that is not discussed in the paper.}. 
The theorem states that 
when testing up to $k$ queries, with up to $B=1$ of those being answered using the private ``Holdout set'',  the size of the holdout set must be at least:
\begin{equation*}
n \geq 96\epsilon^{-2}\ln\left(4k\delta^{-1}\right) \min\{80\sqrt{B\ln(1/\epsilon\delta) \}}, 16B\}
\end{equation*}
Thus, for $k = 10$, $B=1$, $\epsilon = 0.5$, and $\delta = 0.1$: the required sample size is at least
$$n\geq 400\times96\times \ln(400)\min\{80\sqrt{\ln(200)},16\} \geq 3.7\times 10^6.$$
Therefore, even when requesting such, fairly loose, guarantees, $\mathtt{Thresholdout}$ requires an extremely high sample size in order to provide reliable answer to a handful of queries. 

In contrast, we showed in Section~\ref{sec:exper} that \algo{} can provably handle problems with these, and better, parameters while using a holdout set composed by just 4000 samples. Thus \algo{} achieves an improvement of almost three orders of magnitude in terms of \emph{sample complexity} compared to $\mathtt{Thresholdout}$.

Further, $\mathtt{Thresholdout}$ requires that the user specifies \emph{before} the execution the number of queries $k$ which are going to be adaptively chosen to be tested, and the number $B$ of maximum times that the algorithm can tolerate overfit on the public dataset by revealing information from the private ``\emph{holdout}'' dataset. 
These requirements limit the adaptiveness of the process.
 In contrast \algo{} uses the holdout dataset \emph{as much as possible} without a fixed maximum number of queries. Finally, using \Rade{} in evaluating the stopping criterion of the adaptive testing procedure allows \algo{} to evaluate the properties of the actual family of functions tested so far (i.e., their \emph{expressiveness}), rather than just its cardinality, in order to provide guarantees on the quality of the evaluations obtained in the adaptive analysis.

%% file: 06_conclusion.tex
\section{Conclusion}
We presented a rigorous, efficient and practical method for bounding the generalization error in an adaptive sequence of queries tested on the same dataset. While the standard ``rule of thumb'' for responsible data analysis and machine learning is to use a test set only once, our results demonstrate that, with an appropriate control mechanism, it may be possible to use the same test set more than once without significantly reducing the validity of the results. For concreteness, we focused here on the problem of evaluating the expectations of a set of functions in the range $[0,1]$. This problem corresponds to the basic machine learning task of evaluating the correctness of classifiers with bounded $[0,1]$ loss functions. We note that our methods can be extended to a more general setting, for example using new concentration bounds on sub-exponential distributions~\cite{kontorovich2014concentration}, and self-bounding functions~\cite{OnetoGAR13}.